\documentclass[final]{siamltex1213}

\usepackage{overpic}
\usepackage{amsmath,amstext,amsbsy,amssymb}
\usepackage{booktabs,bm}
\usepackage{mathrsfs}
\usepackage{tikz}
\usepackage{booktabs,enumitem}
\usetikzlibrary{positioning}
\usetikzlibrary{shapes,arrows}
\usepackage[fleqn,tbtags]{mathtools}
\usepackage{relsize}

\newcommand{\constant}{\tfrac{1}{2}}

\newcommand{\eps}{\varepsilon}

\newcommand{\argmin}{\mbox{argmin}}

\newcommand{\beq}{\begin{equation}}
\newcommand{\eeq}{\end{equation}}
\renewcommand{\tilde}{\widetilde}
\renewcommand{\hat}{\widehat}
\newcommand{\bit}{\begin{itemize}}
\newcommand{\eit}{\end{itemize}}
\newcommand{\ben}{\begin{enumerate}}
\newcommand{\een}{\end{enumerate}}

\title{Stable extrapolation of analytic functions}
\author{Laurent Demanet \and Alex Townsend\thanks{Department of Mathematics, Massachusetts Institute of Technology, 77 Massachusetts Avenue
Cambridge, MA 02139-4307. This work is supported by National Science Foundation grant No.~1522577. Corresponding author: ajt@mit.edu} }

\begin{document}
\maketitle

\begin{abstract}
This paper examines the problem of extrapolation of an analytic function for $x > 1$ given perturbed samples from an equally spaced grid on $[-1,1]$. Mathematical folklore states that extrapolation is in general hopelessly ill-conditioned, but we show that a more precise statement carries an interesting nuance.  For a function $f$ on $[-1,1]$ that is analytic in a Bernstein ellipse with parameter $\rho > 1$, and for a uniform perturbation level $\eps$ on the function samples, we construct an asymptotically best extrapolant $e(x)$ as a least squares polynomial approximant of degree $M^*$ given explicitly. We show that the extrapolant $e(x)$ converges to $f(x)$ pointwise in the interval $I_\rho\in[1,(\rho+\rho^{-1})/2)$ as $\eps \to 0$, at a rate given by a $x$-dependent fractional power of $\eps$. More precisely, for each $x \in I_{\rho}$ we have
 \[
 |f(x) - e(x)| = \mathcal{O}\left( \eps^{-\log r(x) / \log\rho} \right), \qquad\qquad r(x) = \frac{x+\sqrt{x^2-1}}{\rho},
 \]
up to log factors, provided that the oversampling conditioning is satisfied. That is,
 \[
 M^* \leq \frac{1}{2} \sqrt{N},
 \]
which is known to be needed from approximation theory. In short, extrapolation enjoys a weak form of stability, up to a fraction of the characteristic smoothness length. The number of function samples, $N+1$, does not bear on the size of the extrapolation error provided that it obeys the oversampling condition.
We also show that one cannot construct an asymptotically more accurate extrapolant from $N+1$ equally spaced samples than $e(x)$, using any other linear or nonlinear procedure. The proofs involve original statements on the stability of polynomial approximation in the Chebyshev basis from equally spaced samples and these are expected to be of independent interest. 
\end{abstract}



\section{Introduction}\label{sec:Introduction} 


Stable extrapolation is a topic that has traditionally been avoided in numerical analysis, perhaps out of a concern that positive results may be too weak to be interesting. The thorough development of approximation theory for $\ell_1$ minimization over the past ten years; however, has led to the discovery of new regimes where {\em interpolation} of smooth functions is accurate, under a strong assumption of Fourier sparsity~\cite{CRT04}. More recently, these results have been extended to deal with the {\em extrapolation} case, under the name super-resolution~\cite{CFG12, Demanet_15_01}. This paper seeks to bridge the gap between these results and traditional numerical analysis, by rolling back the Fourier-sparse assumption and establishing tight statements on the accuracy of extrapolation under the basic assumption that the function is analytic and imperfectly known at equally spaced samples.

\subsection{Setup}

A function $f:[-1,1]\rightarrow\mathbb{C}$ is real-analytic when each of its Taylor expansions, centered at each point $x$, converges in a disk of radius $R > 0$. While the parameter $R$ is one possible measure of the smoothness of $f$, we prefer in this paper to 
consider the largest Bernstein ellipse, in the complex plane, to which $f$ can be analytically continued. We say that a function $f:[-1,1]\rightarrow\mathbb{C}$ is {\em analytic with a Bernstein parameter $\rho>1$} if it is analytically continuable to a function that is analytic in the open 
ellipse with foci at $\pm1$, semiminor and semimajor axis 
lengths summing to $\rho$, denoted by $E_{\rho}$, and bounded in $E_\rho$ so that
$|f(z)|\leq Q$ for $z\in E_\rho$ and $Q<\infty$.\footnote{The relationship between $R$ and $\rho$ is found by considering $f$ analytic in the so-called 
stadium of radius $R>0$, i.e., the region $S_R = \{z\in\mathbb{C}: \inf_{x\in[-1,1]} |z-x|< R\}$. 
If $f$ is analytic with a Bernstein parameter $\rho>1$, then $f$ is also analytic in 
the stadium with radius $R = (\rho+\rho^{-1})/2-1$. Conversely, if $f$ is analytic in $S_R$, then $f$ is analytic with a 
Bernstein parameter $\rho = R+\sqrt{R^2+1}$. See~\cite{demanet-cheb,demanet-fio} for details.} We denote the set of such functions as $B_\rho(Q)$.

Such a function $f$ has a unique, bounded analytic continuation in the interval $I_\rho = [1,(\rho+\rho^{-1})/2)$, which serves as the reference for measuring the extrapolation error. We denote by $r(x)$, or simply $r$, the nondimensional length parameter in this interval,
\[
r = \frac{x+\sqrt{x^2-1}}{\rho},
\]
so that $\frac{1}{\rho} \leq r < 1$ for $x \in I_\rho$.

The question we answer in this paper is: ``How best to stably extrapolate an analytic function from imperfect equally spaced samples?"
More precisely, for known parameters $N$, $\rho$, $\eps$, and $Q$ we assume that
\begin{itemize}[leftmargin=*]
\item $f \in B_\rho(Q)$;
\item $N+1$ imperfect equally spaced function samples of $f$ are given. That is, the vector $f(\underline{x}^{equi})+\underline{\eps}$ is known, 
where $\underline{x}^{equi}$ is the vector of $N+1$ equally spaced points on $[-1,1]$ so that $x_k = 2k/N-1$ for $0\leq k\leq N$ 
and $\underline{\eps}$ is a perturbation vector with $\|\underline{\eps}\|_\infty \leq \eps$; and 
\item $x\in I_\rho$ is an extrapolation point, where $I_\rho = [1,(\rho+\rho^{-1})/2)$. 
\end{itemize}
Our task is to construct an extrapolant $e(x)$ for $f(x)$ in the interval $I_\rho$ from the imperfect equally spaced samples 
that minimizes the extrapolation error $|f(x)-e(x)|$ for $x\in I_\rho$.

Extrapolation is far from being the counterpoint to interpolation, and several different ideas are required. First, the polynomial interpolant of an analytic function $f$ at 
$N+1$ equally spaced points on $[-1,1]$ can suffer from wild oscillations near $\pm 1$, known as  Runge's phenomenon~\cite{Runge_01_01}. Second, the construction 
of an equally spaced polynomial interpolant is known to be exponentially ill-conditioned, leading to practical problems with 
computations performed in floating point arithmetic. Various remedies are proposed for the aforementioned problems,\footnote{Among them, least squares polynomial fitting~\cite{Cohen_13_01}, mock Chebyshev interpolation~\cite{Boyd_09_01}, polynomial overfitting with constraints~\cite{Boyd_92_01}, and the Bernstein polynomial basis~\cite[Sec.~6.3]{Powell_81_01}. For an extensive list, see~\cite{Platte_11_01}.} and in this paper we show that one approach is simply least-squares approximation by polynomials of much lower degree than the number of function samples.

For a given integer $0 \leq M \leq N$, we denote by $p_M(x)$ the least squares polynomial fit of degree $M$ to the imperfect samples, i.e.,
\begin{equation} 
p_M = \argmin_{p \in \mathcal{P}_M} \|  f(\underline{x}^{equi}) + \underline{\eps} - p(\underline{x}^{equi}) \|_2,
\label{eq:leastSquares}
\end{equation}
where $\mathcal{P}_M$ is the space of polynomials of degree at most $M$.
In this paper, we show that a near-best extrapolant $e(x)$ is given by
\beq\label{eq:e(x)}
e(x) = p_{M^*}(x),
\eeq
where
\beq\label{eq:M}
M^* = \Bigg\lfloor \min \left\{ \frac{1}{2} \sqrt{N}, \frac{\log ( Q / \eps )}{\log(\rho)} \right\} \Bigg\rfloor.
\eeq
Here, $\lfloor a\rfloor$ denotes the largest integer less than or equal to $a$, but exactly how the integer part is 
taken in~\eqref{eq:M} is not particularly important. The formula for $M^*$ in~\eqref{eq:M} is derived by approximately 
balancing two terms: a noiseless term that is geometrically decaying to zero with $M$ and a noise term that 
is exponentially growing with $M$. It is the exponentially growing noise term that has lead researchers to the conclusion that polynomial 
extrapolation is unstable in practice. The balance of these two terms roughly minimizes the extrapolation error. 
If $\log ( Q / \eps )/\log(\rho) < \tfrac{1}{2} \sqrt{N}$, then this balancing can be achieved without 
violating a necessary oversampling condition; otherwise, $\log ( Q / \eps )/\log(\rho) \geq \tfrac{1}{2} \sqrt{N}$
and one gets as close as possible to the balancing of the two terms by setting $M^* =  \tfrac{1}{2} \sqrt{N}$. 


\subsection{Main results}

The behavior of the extrapolation error depends on whether $M^*=\tfrac{1}{2}\sqrt{N}$ or not (see~\eqref{eq:M}), 
and the two corresponding regimes are referred to as {\em undersampled} and {\em oversampled}, respectively. 

\begin{definition} The extrapolation problem with parameters $(N, \rho, \eps, Q)$ is said to be oversampled if
 \beq\label{eq:sampling}
\frac{\log ( Q / \eps )}{\log(\rho)} < \frac{1}{2} \sqrt{N}.
 \eeq
Conversely, if this inequality is not satisfied, then the problem is said to be undersampled. 
 \end{definition}
%
 
The relation between $M$ and $N$ stems from the observation that polynomial approximation on an equally spaced grid can be computed stably when $M \leq \frac{1}{2} \sqrt{N}$, as we show in the sequel, but not if $M$ is asymptotically larger than $\smash{\sqrt{N}}$~\cite[p.~3]{Platte_11_01}. In~\cite{Dahlquist_03_01} it is empirically
observed that~\eqref{eq:leastSquares} can be solved without any numerical issues if $M<2\sqrt{N}$ and yet another illustration of this relationship is 
the so-called mock-Chebyshev grid, which is a subset of an $N+1$ equally spaced grid of size $M \sim \sqrt{N}$ that allows for 
stable polynomial interpolation~\cite{Boyd_09_01}. 

We now give one of our main theorems.  For convenience, let 
\[
\alpha(x) = - \frac{\log r(x)}{\log\rho},
\]
which is the fractional power of the perturbation level $\eps$ in the error bound below.

\begin{theorem}\label{teo:main} Consider the extrapolation problem with parameters $(N, \rho, \eps, Q)$.
\begin{itemize}[leftmargin=*]
\item If (\ref{eq:sampling}) holds (oversampled case), then for all $x \in I_\rho$,
\beq\label{eq:oversampled}
|f(x) - e(x)| \leq C_{\rho, \eps} \; \frac{Q}{1-r(x)} \; \left( \frac{\eps}{Q} \right)^{\alpha(x)},
\eeq
where $C_{\rho, \eps}$ is a constant that depends polylogarithmically on $1/\eps$. 
\item If (\ref{eq:sampling}) does not hold (undersampled case), then for all $x \in I_\rho$,
\beq\label{eq:undersampled}
|f(x) - e(x)| \leq C_{\rho, N} \; \frac{Q}{1-r(x)} \; r(x)^{\frac{1}{2} \sqrt{N}},
\eeq
where $C_{\rho, N}$ is a constant that depends polynomially on $N$. 
\end{itemize}
\end{theorem}

Note that $\alpha(x)$ is strictly decreasing in $x\in I_\rho$ with $\alpha(1) = 1$ (the error is proportional to $\eps$ at $x=1$, as expected) to $\alpha( (\rho+\rho^{-1})/2) = 0$ where the Bernstein ellipse meets the real axis (there is no expectation of control over the extrapolation error at $x = (\rho+\rho^{-1})/2$ since $f$ could be a rational function with a pole outside the Bernstein ellipse). For $1 < x < (\rho+\rho^{-1})/2$, it is surprising that the minimum extrapolation error is not proportional to $\eps$ itself, but an $x$-dependent fractional power of it. Note that the factor $1/(1 - r(x))$ also blows up at the endpoint at $x = (\rho+\rho^{-1})/2$.
Figure~\ref{fig:alphaPlots} (left) shows the fractional power of $\eps$ that is achieved by our extrapolant in the oversampled case and 
Figure~\ref{fig:alphaPlots} (right) shows the bound in~\eqref{eq:oversampled} without the constants for extrapolating the function $1/(1+x^2)$ in double 
precision. 

\begin{figure} 
\centering 
\begin{minipage}{.49\textwidth}
\begin{overpic}[width=\textwidth]{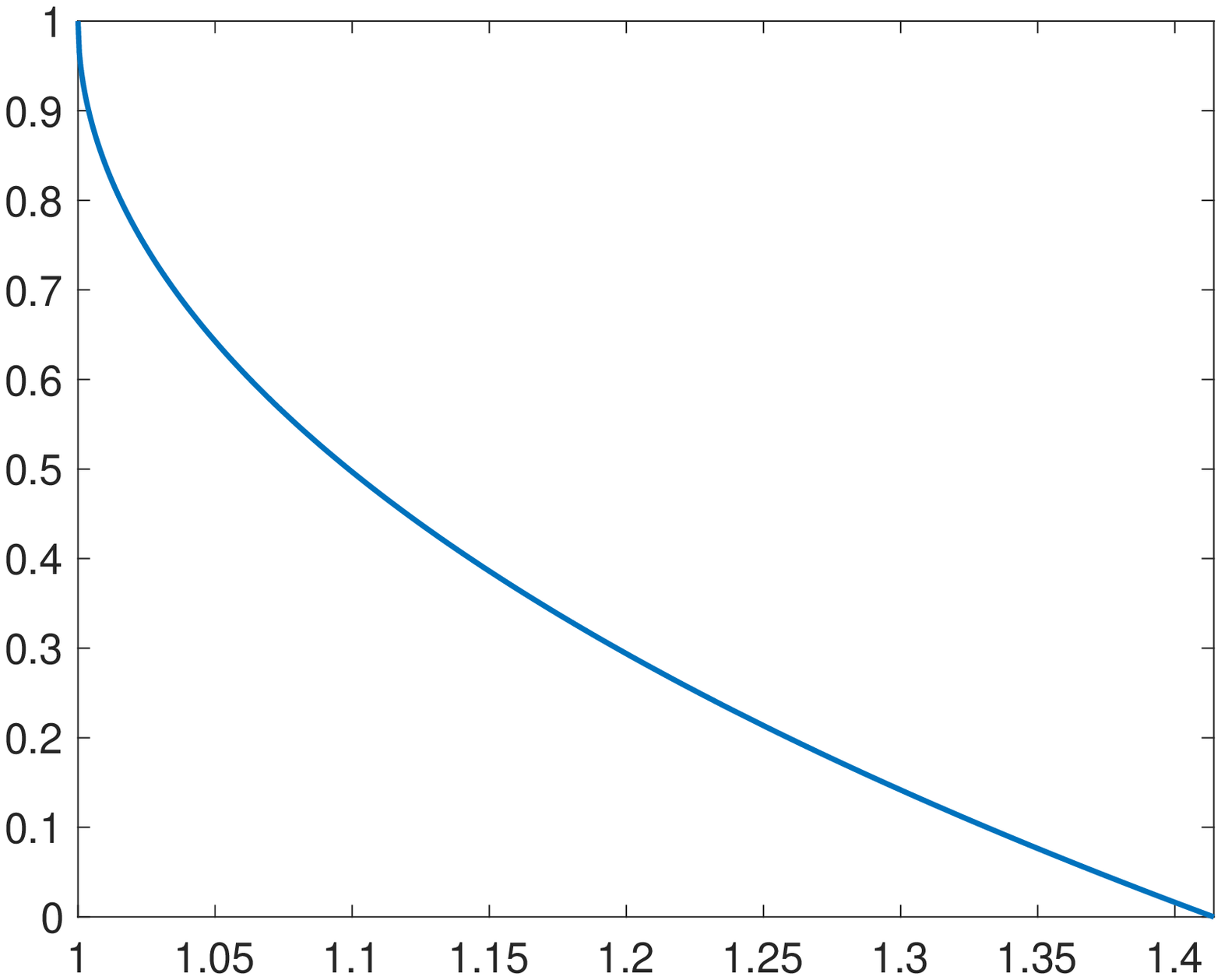} 
\put(50,0) {$x$} 
\put(27, 58) {$\alpha(x) = -\log r(x)/\log\rho$}
\end{overpic} 
\end{minipage}
\begin{minipage}{.49\textwidth}
\begin{overpic}[width=\textwidth]{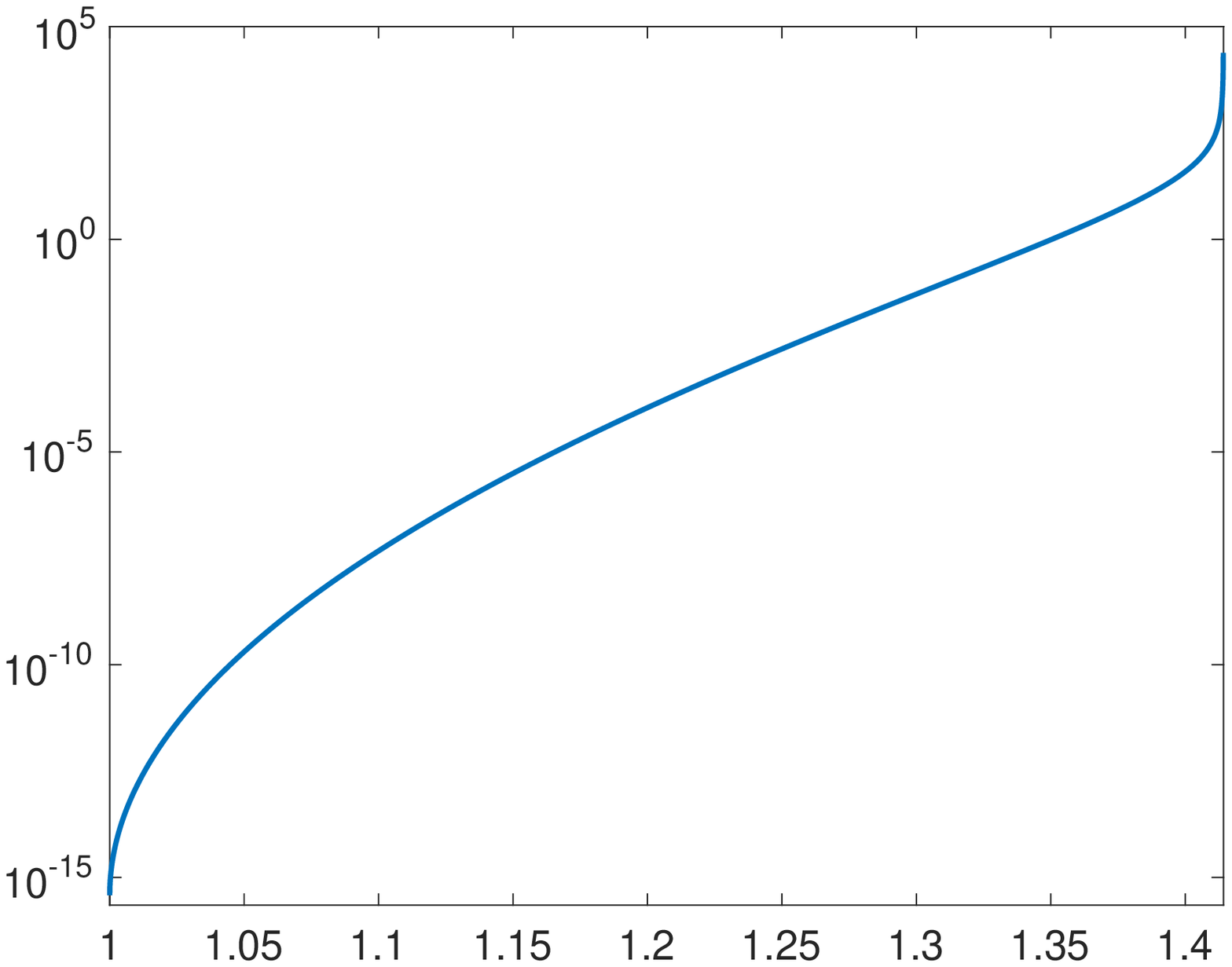} 
\put(50,0) {$x$} 
\put(43, 58) {$\frac{\eps^{-\alpha(x)}}{1-r(x)}$}
\end{overpic} 
\end{minipage}
\caption{In the oversampled case, the near-optimal extrapolant for $f(x)$ in $I_\rho = [1,(\rho+\rho^{-1})/2)$ is given by $e(x) = p_{M^*}(x)$, where $M^* = \lfloor \log(Q/\eps)/\log\rho \rfloor$. The accuracy of extrapolation depends on a fractional power of $\eps$ multiplied by $\smash{(1-r(x))^{-1}}$, i.e., $\eps^{-\alpha(x)}/(1-r(x))$, where $\alpha(x) = \log r(x)/\log\rho$. Here, $\alpha(x)$ (left) and 
the extrapolation error with the constant (right) is shown for the function $f(x) = 1/(1+x^2)$, with $\rho = 1+\sqrt{2}$ ($f\in B_{\rho'}(Q')$ for any $\rho'<1+\sqrt{2}$), and $\eps = 2.2\times 10^{-16}$. In the oversampled case, no linear or nonlinear scheme can provide an asymptotically more accurate extrapolant in general than this bound (see Proposition~\ref{teo:minimax}).}
\label{fig:alphaPlots} 
\end{figure} 

The bound~\eqref{eq:oversampled} in Theorem~\ref{teo:main} cannot be meaningfully improved, as the following proposition shows.

\begin{proposition}\label{teo:minimax}
Consider the extrapolation problem with parameters $(N, \rho, \eps, Q)$ such that~\eqref{eq:sampling} holds. Then, there exists a function $g \in B_{\rho'}(Q')$ for all $\rho' < \rho$ such that 
\[
\max_{x \in [-1,1]} | g(x) | \leq \eps,
\]
and, for $x \in  I_\rho$ and some $c_\rho > 0$,
\[
| g(x) | \geq c_\rho \; \frac{1}{1 - r(x)} \; \eps^{\alpha(x)}.
\]
\end{proposition}

In other words, $g(x)$ is a valid extrapolant to $f(x) = 0$, to within a tolerance of $\eps$ on $x \in [-1,1]$, yet it departs 
from zero at the same asymptotic rate as the upper bound in Theorem \ref{teo:main} for $x>1$.  This means that there is no other linear or nonlinear 
procedure for constructing an extrapolant from samples on $[-1,1]$ that can do asymptotically better than the extrapolant that we 
construct in Theorem~\ref{teo:main}. For example, an extrapolant constructed by Chebyshev interpolation, piecewise polynomials, rational functions, or any other linear
or nonlinear procedure cannot deliver an extrapolation error that is better than~\eqref{eq:oversampled} in any meaningful way. 

%


\subsection{Discussion}

The number of equally spaced function samples $N+1$ separates two important regimes:
\begin{itemize}[leftmargin=*]
\item {\em Oversampled regime.} If $N$ is sufficiently large that~\eqref{eq:sampling} holds, then further refining of the grid does not improve the extrapolation error. In this regime it is the value of $\eps$ that dictates the error~\eqref{eq:oversampled}. The problem is essentially one of (deterministic) statistics.
\item {\em Undersampled regime.} If $\eps$ is sufficiently small that~\eqref{eq:sampling} does not hold, then the accuracy of the extrapolant is mostly blind to the fact that there is a perturbation level at all. In this regime, it is the number of function samples that dictates the error~\eqref{eq:undersampled}. The problem is essentially one of (classical) numerical analysis.
\end{itemize}

A similar phenomenon appears in the related problem of super-resolution from bandlimited measurements, where it is also the perturbation level of the function samples that determines the recovery error, provided the number of samples is above a certain threshold~\cite{Demanet_13_01,Demanet_15_01}.

In the oversampled case, there exists a perturbation vector for which the actual extrapolation error nearly matches the error bound for the proposed extrapolant $e(x)$ in (\ref{eq:e(x)}). This implies that $e(x)$ is a {\em minimax} estimator for $f(x)$, in the sense that it nearly attains the best possible error
\[
E_{\mbox{\scriptsize minmax}}(x) = \inf_{\hat{e}} \sup_{f, \underline{\eps}} |f(x) - \hat{e}(x)|,
\]
where the infimum is taken over all possible mappings from the perturbed samples to functions of $x \in I_\rho$, and the supremum assumes that $f \in B_{\rho}(Q)$ and $\| \underline{\eps} \|_{\infty} \leq \eps$. This paper does not address the question of whether $e(x)$ is also minimax in the undersampled case.

The statement that ``the value of $N$ does not matter provided it is sufficiently large" should not be understood as ``acquiring more function samples does not matter for extrapolation". The threshold phenomenon is specific to the model of a deterministic perturbation of level $\eps$, which is independent of $N$. If instead the entries of 
the perturbation vector $\underline{\eps}$ are modeled as independent and identically distributed Gaussian entries, $\mathcal{N}(0,s^2)$, then the approximation and extrapolation errors include an extra factor $1/\sqrt{N}$, linked to the local averaging implicitly performed in the least-squares polynomial fits. In this case the extrapolant converges pointwise to $f$ as $N\rightarrow\infty$, though only at the so-called parametric rate expected from statistics, not at the subexponential rate~\eqref{eq:undersampled} expected from numerical analysis (see Section~\ref{sec:ExtrapolationWithNoise}).

\subsection{Auxiliary results of independent interest}
Before we can begin to analyze how to extrapolate analytic functions, we derive results regarding the conditioning and approximation power of least squares approximation as well as its robustness to perturbed function samples. These 
results become useful in Section~\ref{sec:Extrapolation} for 
understanding how to do extrapolation successfully. 

Our auxiliary results may be independent interest so we summarize them here:
\begin{itemize}[leftmargin=*]
 \item Theorem~\ref{thm:LegendreVandermondeSingularValues}: The condition number of the rectangular $(N+1)\times (M+1)$ Legendre--Vandermonde 
 matrix at equally spaced points (see~\eqref{eq:LegendreVandermondeMatrix}) with $M\leq \constant\sqrt{N}$ is bounded by $\sqrt{5(2M+1)}$.
  \item Theorem~\ref{thm:ChebyshevSingularValues}: The condition number of the rectangular $(N+1)\times (M+1)$ Chebyshev--Vandermonde 
 matrix at equally spaced points (see~\eqref{eq:ChebyshevVandermondeMatrix}) with $M\leq \constant\sqrt{N}$ is bounded by $\sqrt{375(2M+1)/2}$.
 \item Theorem~\ref{thm:AnalyticLeastSquares}: When $M\leq \constant\sqrt{N}$, $\|f-p_M\|_\infty = \sup_{x\in[-1,1]}|f(x)-p_M(x)|$ converges geometrically to zero 
 as $M\rightarrow\infty$.
 \item Corollary~\ref{cor:ApproximationPowerWithNoise}: When $M\leq \constant\sqrt{N}$ is fixed and the function samples from $f$
 are perturbed by Gaussian noise with a variance of $s^2$, the expectation of $\|f-p_M\|_\infty$ converges to zero as $N\rightarrow \infty$ 
 like $\mathcal{O}(s/\sqrt{N})$.
 \item Theorem~\ref{thm:extrapolation}: When $M\leq \constant\sqrt{N}$ and the function samples are noiseless the 
 extrapolation error $|f(x)-p_M(x)|$ for each $x\in I_\rho$ converges geometrically to zero as $M\rightarrow\infty$. 
 \item Corollary~\ref{cor:NoisyErrorBound}: If one exponentially oversamples on $[-1,1]$, i.e., 
 $M\leq c\log(N)$ for a small constant $c$ and the function samples are perturbed by Gaussian noise, 
 then $|f(x)-p_M(x)|$ converges to zero as $M\rightarrow\infty$ for each $x\in I_\rho$.
 \end{itemize}
 
Note that Theorem~\ref{thm:AnalyticLeastSquares} shows that the convergence of $p_M(x)$ is geometrically 
fast with respect to $M$, but subexponential with respect in $N$ when $M = \lfloor\constant\sqrt{N}\rfloor$. One cannot achieve
a better convergence rate with respect to $N$ by using any other stable linear or nonlinear approximation scheme based on equally spaced 
function samples~\cite{Platte_11_01}.

Readers familiar with the paper by Adcock and Hansen~\cite{Adcock_12_01}, which shows how to
stably recover functions from its Fourier coefficients may consider 
Section~\ref{sec:EquallySpacedLeastSquares} and Section~\ref{sec:ApproximationPower} as a 
discrete and nonperiodic analogue of their work.  Related work based on Fourier expansions, 
includes the recovery of piecewise analytic functions from Fourier modes~\cite{Adcock_15_01} 
and a detailed analysis of the stability barrier in~\cite{Adcock_14_01}.  

%
%

\subsection{Notation and background material}\label{sec:NotationBackground}  
The polynomial $p_M(x)$ in~\eqref{eq:leastSquares} can be represented in any 
polynomial basis for $\mathcal{P}_M$. We use the 
Chebyshev polynomial basis because it is convenient for practical computations. 
That is, we express $p_M(x)$ in a Chebyshev expansion given by 
\begin{equation}
 p_M(x) = \sum_{k=0}^M c_k^{cheb} T_k(x), \qquad  T_k(x) = \cos(k\cos^{-1} x ), \qquad x\in[-1,1],
\label{eq:ChebyshevExpansion}
\end{equation} 
where $T_k$ is the Chebyshev polynomial of degree $k$, and we seek
the vector of Chebyshev coefficients $\underline{c}^{cheb} = (c_0^{cheb},\ldots,c_M^{cheb})^T$ 
so that $p_M(x)$ minimizes the $\ell_2$-norm in~\eqref{eq:leastSquares}.

The vector of Chebyshev coefficients $\underline{c}^{cheb}$ for $p_M(x)$ in~\eqref{eq:leastSquares} satisfies the so-called {\em normal 
equations}~\cite[Alg.~5.3.1]{Golub_96_01} written as
\begin{equation}
 \mathbf{T}_M(\underline{x}^{equi})^*\mathbf{T}_M(\underline{x}^{equi})\underline{c}^{cheb} = \mathbf{T}_M(\underline{x}^{equi})^*\left(\underline{f}+\underline{\eps}\right),
\label{eq:normalEquations}
\end{equation} 
where $\underline{f} = f(\underline{x}^{equi})$ is the vector of equally spaced
samples and $\mathbf{T}_M(\underline{x}^{equi})$ denotes the $(N+1)\times (M+1)$ 
Chebyshev--Vandermonde\footnote{The Chebyshev--Vandermonde matrix in~\eqref{eq:ChebyshevVandermondeMatrix} is 
the same as the familiar Vandermonde matrix except the monomials are 
replaced by Chebyshev polynomials.} matrix, 
\begin{equation}
 \mathbf{T}_M(\underline{x}^{equi}) =\begin{bmatrix} 
 T_0(x_0^{equi}) & \cdots & T_M(x_0^{equi}) \cr 
 \vdots & \ddots & \vdots \cr 
 T_0(x_N^{equi}) & \cdots & T_M(x_N^{equi}) 
\end{bmatrix}.
\label{eq:ChebyshevVandermondeMatrix}
\end{equation} 
This converts~\eqref{eq:leastSquares} into a routine linear 
algebra task that can be solved by Gaussian elimination and 
hence, the computation of $p_M(x)$ in~\eqref{eq:leastSquares} is simple.

If $f$ is analytic with a Bernstein parameter $\rho$,
then $f$ has a Chebyshev expansion $f(x) = \sum_{n=0}^{\infty} a_n^{cheb}T_n(x)$ for $x\in[-1,1]$ with 
coefficients that decay geometrically to zero as $n\rightarrow\infty$.
\begin{proposition} 
 Let $f$ be analytic with a Bernstein parameter $\rho>1$ and $Q<\infty$. Then,
 there are coefficients $a^{cheb}_n$ for $n\geq 0$ such that
\begin{itemize}[leftmargin=*]
 \item $f(x) = \sum_{n=0}^{\infty} a_n^{cheb}T_n(x)$, where the series converges uniformly and absolutely to $f$,
 \item $|a_0^{cheb}|\leq Q$ and $|a_n^{cheb}|\leq 2Q\rho^{-n}$ for $n\geq 1$, and
 \item $\sup_{x\in[-1,1]}|f(x) - f_N(x)| \leq 2Q\rho^{-N}/(\rho-1)$, where $f_N(x) = \sum_{n=0}^{N} a_n^{cheb}T_n(x)$ and $N\geq 0$.  
\end{itemize}
\label{prop:AnalyticConvergence} 
\end{proposition}
\begin{proof} 
 See~\cite[Thm.~8.1]{Trefethen_13_01} and~\cite[Thm.~8.2]{Trefethen_13_01}.
\end{proof}

Proposition~\ref{prop:AnalyticConvergence} says that the degree $N$ 
polynomial $f_N$, constructed by truncating the Chebyshev expansion of 
$f$, converges geometrically to $f$. In general, one cannot expect faster 
convergence for polynomial approximants of analytic 
functions. 
However, it is rare in practical applications for the 
Chebyshev expansion of $f$ to be known in advance. Instead, one 
usually attempts to emulate the degree $N$ 
polynomial $f_N$ by a polynomial interpolant constructed 
from $N+1$ samples of $f$. 
When the samples are taken from Chebyshev points or 
Gauss--Legendre nodes on $[-1,1]$ a polynomial interpolant can 
be constructed in a fast and stable manner~\cite{Driscoll_14_01,Hale_15_01}.
The same cannot be said for equally spaced samples
on $[-1,1]$~\cite{Platte_11_01}. In this paper we explore the least squares
polynomial approximation as a practical alternative to polynomial interpolation when 
equally spaced samples are known. 

For the convenience of the reader we summarize our main notation in Table~\ref{tab:notation}. 
\begin{table}
\begin{tabular}{cl}
\toprule
Notation & Description \cr 
\midrule
 $B_\rho(Q)$&  A function $f$ that is analytic in $E_\rho$ and $|f(z)|\leq Q$ for $z\in E_\rho$,\\[5pt]
 &  where $E_\rho$ is the region enclosed by a ellipse with foci at $\pm1$ and\\[5pt]
 & semimajor and semiminor axis lengths summing to $\rho$ \\[5pt]
 $f$ & An analytic function on $[-1,1]$ with Bernstein parameter $\rho>1$\\[5pt]
 $N+1$ & The number of equally spaced function samples from $[-1,1]$\\[5pt]
 $M$ & The desired degree of a polynomial approximation to $f$\\[5pt]
 $p_M$ & The least squares polynomial approximation of $f$, see~\eqref{eq:leastSquares}\\[5pt]
 $T_k(x)$ &  Chebyshev polynomial (1st kind) of degree $k$ \\[5pt]
 $P_k(x)$ &  Legendre polynomial of degree $k$ \\[5pt]
 $\underline{x}^{equi}$ & Vector of equally spaced points on $[-1,1]$, i.e., $x_k^{equi}=2k/N-1$\\[5pt]
 $\underline{f}$, $f(\underline{x}^{equi})$ & Vector of equally spaced function samples of $f$ \\[5pt]
 $\underline{\eps},\eps$ & Vector of perturbations in the function samples of $f$, $\|\underline{\eps}\|_\infty\leq \eps$ \\[5pt]
 $\mathbf{T}_M(\underline{x})$ & The matrix $\begin{bmatrix} 
 T_0(x_0) & \cdots & T_M(x_0) \cr 
 \vdots & \ddots & \vdots \cr 
 T_0(x_N) & \cdots & T_M(x_N) 
\end{bmatrix} \in\mathbb{R}^{(N+1)\times (M+1)}$ \\[5pt]
$\Lambda_N(\underline{x})$ & Lebesgue constant of $x_0,\ldots,x_N$, see Definition~\ref{def:Lebesgue} \\[5pt]
$S$ & Change of basis matrix from Legendre to Chebyshev coefficients\\[5pt]
& $S_{ij} = \begin{cases}\frac{1}{\pi}\Psi\left(\frac{j}{2}\right)^2, & 0=i\leq j\leq M, \text{ }j \text{ even},\cr \frac{2}{\pi}\Psi\left(\frac{j-i}{2}\right)\Psi\left(\frac{j+i}{2}\right),& 0<i\leq j\leq M, \text{ } i+j \text{ even}, \cr 0, & \text{otherwise},\end{cases}$\\[5pt] 
&  where $\Psi(i) = \Gamma(i + 1/2)/\Gamma(i+1)$ and $\Gamma(x)$ is the Gamma function\\[5pt]
$\sigma_{k}(A)$ & The $k$th largest singular value of the matrix $A$ \\[5pt]
$\kappa_2(A)$ & The $2$-norm condition number given by $\kappa_2(A) = \|A\|_2\|A^{-1}\|_2$\\[5pt]
$\mathcal{N}(\mu,s^2)$ & Gaussian distribution with mean $\mu$ and variance $s^2$ \\[5pt]
$\mathbb{E}[X]$ & The expectation of the random variable $X$ \\[5pt]
\bottomrule
\end{tabular}
\caption{A summary of our notation.}
\label{tab:notation} 
\end{table} 

\subsection{Structure of the paper} 
The paper is structured as follows. In Section~\ref{sec:EquallySpacedInterpolation} 
we further investigate the exponential ill-conditioning associated to polynomial 
interpolation.  In Section~\ref{sec:EquallySpacedLeastSquares} we
show that the normal equations associated with~\eqref{eq:leastSquares} are well-conditioned.  
In Section~\ref{sec:ApproximationPower}
we prove that for analytic functions the least squares polynomial fit is asymptotically optimal 
for a well-conditioned linear approximation scheme when $M\leq\constant\sqrt{N}$ and in Section~\ref{sec:StableLeastSquares} we show that it is 
also robust to noisy function samples. In Section~\ref{sec:Extrapolation} we show 
that the solution $p_M(x)$ from~\eqref{eq:leastSquares} can be used to extrapolate outside
of $[-1,1]$ if significant care is taken and we construct the asymptotically best extrapolant $e(x)$
as a polynomial. Finally, in Section~\ref{sec:LeastSquaresAlgorithm} we describe a direct 
algorithm for solving~\eqref{eq:leastSquares} in $\mathcal{O}(M^3+MN)$ operations based on 
Toeplitz and Hankel matrices. 

\section{How bad is equally spaced polynomial interpolation?}\label{sec:EquallySpacedInterpolation}
First, we explore how bad equally spaced polynomial interpolation is in practice by 
taking $M = N$ in~\eqref{eq:leastSquares} and showing that the condition number 
of the $(N+1)\times (N+1)$ Chebyshev--Vandermonde matrix $\mathbf{T}_N(\underline{x}^{equi})$ in~\eqref{eq:ChebyshevVandermondeMatrix} grows exponentially with $N$. 

When $M = N$ the polynomial $p_M(x)$ that minimizes the $\ell_2$-norm in~\eqref{eq:leastSquares} also 
interpolates $f$ at $\underline{x}^{equi}$ and 
the vector of Chebyshev coefficients $\underline{c}^{cheb}$ for 
$p_M(x)$ in~\eqref{eq:ChebyshevExpansion} satisfies the linear system
\begin{equation}
 \mathbf{T}_N(\underline{x}^{equi})\underline{c}^{cheb} = \left(\underline{f}+\underline{\eps}\right).
\label{eq:linearSystem}
\end{equation} 
By the Lagrange interpolation theorem, $\mathbf{T}_N(\underline{x}^{equi})$ 
is invertible and mathematically there is a unique solution vector $\underline{c}^{cheb}$ to~\eqref{eq:linearSystem}.
Unfortunately, it turns out that $\mathbf{T}_N(\underline{x}^{equi})$ is exponentially close to 
being singular and the vector $\underline{c}^{cheb}$ is far too sensitive to the
perturbations in $\underline{f}+\underline{\eps}$ for~\eqref{eq:linearSystem} to be of practical use when $N$ is large. 

We explain why the condition number of $\mathbf{T}_N(\underline{x}^{equi})$ grows
exponentially with $N$ by relating it to the poorly behaved 
{\em Lebesgue constant} of $\underline{x}^{equi}$. 
\begin{definition}[Lebesgue constant]
 Let $x_0,\ldots,x_N$ be a set of $N+1$ distinct points in $[-1,1]$. Then, the 
 Lebesgue constant of $\underline{x}=(x_0,\ldots,x_N)^T$ is defined by 
 \begin{equation}
  \Lambda_N(\underline{x}) = \sup_{x\in[-1,1]}\sum_{j=0}^N|\ell_j(x)|,\qquad \ell_j(x) = \prod_{k=0,k\neq j}^N \frac{x-x_k}{x_j-x_k}.
 \label{eq:Lebesgue}
 \end{equation} 
 \label{def:Lebesgue}
\end{definition}

To experts the fact that $\Lambda_N(\underline{x})$ and the condition number of 
$\mathbf{T}_N(\underline{x}^{equi})$ are related is not too surprising 
because polynomial interpolation is a linear approximation scheme~\cite{Platte_11_01}.
However, the Lebesgue constant $\Lambda_N(\underline{x})$ is usually interpreted 
as a number that describes how good
polynomial interpolation of $f$ at $x_0,\ldots,x_N$ is in comparison to the 
best minimax polynomial approximation of degree $N$. That is, the polynomial 
interpolant of $f$ at $\underline{x}$ is suboptimal by a factor of at 
most $1+\Lambda_N(\underline{x})$~\cite[p.~24]{Powell_81_01}. Using 
$\|\cdot\|_\infty$ to denote the absolute maximum norm of a function on $[-1,1]$, this
can be expressed as
\[
 \|f - p_N \|_\infty \leq \left(1 + \Lambda_N(\underline{x})\right)\inf_{q\in\mathcal{P}_M}\|f - q\|_\infty, 
\]
where $p_N$ is the polynomial of degree at most $N$ such that $p_N(x_k) = f(x_k)$ for $0\leq k\leq N$. 
For example, when the interpolation nodes are the Chebyshev points (of the first kind), i.e., 
\begin{equation}
 x_k^{cheb} = \cos((k+1/2)\pi/(N+1)),\qquad  0\leq k\leq N,
\label{eq:ChebyshevPoints}
\end{equation} 
the Lebesgue constant $\Lambda_N(\underline{x}^{cheb})$ grows modestly with $N$ and
is bounded by $\tfrac{2}{\pi}\log(N+1)+1$~\cite{Brutman_96_01}. Thus, the polynomial interpolant 
of $f$ at $\underline{x}^{cheb}$ is near-best (off by at most a logarithmic factor).
In addition, we have\footnote{To show that 
$\kappa_2\left(\mathbf{T}_N(\underline{x}^{cheb})\right) = \sqrt{2}$, note that 
$\mathbf{T}_N(\underline{x}^{cheb})$ is the discrete cosine transform (of type III)~\cite{Strang_99_01}. Thus, 
$\mathbf{T}_N(\underline{x}^{cheb})D^{-1/2}$ is an orthogonal matrix with
$D = {\rm diag}(N+1,(N+1)/2,\ldots,(N+1)/2)$.} 
$\kappa_2\left(\mathbf{T}_N(\underline{x}^{cheb})\right) = \sqrt{2}$. 
This means that polynomial interpolants at Chebyshev points are a powerful
tool for approximating functions even when polynomial degrees are in the thousands or 
millions~\cite{Driscoll_14_01}. 

In stark contrast, the Lebesgue constant for equally spaced points 
explodes exponentially with $N$ and we have~\cite[Thm.~2]{Trefethen_91_01}
\[
 \frac{2^{N-2}}{N^2} < \Lambda_N(\underline{x}^{equi}) < \frac{2^{N+3}}{N}.
\]
Therefore, an equally spaced polynomial interpolant of
$f$ can be exponentially worse than the best minimax polynomial approximation of the same degree. 
Moreover, in Theorem~\ref{thm:illconditioned} we show that 
$\kappa_2\left(\mathbf{T}_N(\underline{x}^{equi})\right)$ is related 
to $\Lambda_N(\underline{x}^{equi})$ and grows at an exponential rate, making practical computations
in floating point arithmetic difficult.

\begin{theorem} 
Let $\underline{x} = (x_0,\ldots,x_N)$ be a vector of $N+1$ distinct points on $[-1,1]$. Then,
\[
\Lambda_N(\underline{x}) \leq \kappa_2 \left( \mathbf{T}_N(\underline{x}) \right) \leq \sqrt{2}(N+1) \Lambda_N(\underline{x}),
\]
where $\kappa_2$ is the $2$-norm condition number of a matrix, 
$\Lambda_N(\underline{x})$ is the Lebesgue constant of $\underline{x}$, and
$\ell_j(x)$ for $0\leq j\leq N$ is given in~\eqref{eq:Lebesgue}.  
\label{thm:illconditioned}
\end{theorem}
\begin{proof} 
The vector $\underline{x}$ contains $N+1$ distinct points so that $\mathbf{T}_N(\underline{x})$ is an invertible matrix. 
We write $\kappa_2\left( \mathbf{T}_N(\underline{x}) \right) = \| \mathbf{T}_N(\underline{x})\|_2\| \mathbf{T}_N(\underline{x})^{-1}\|_2$ 
and proceed by bounding $\| \mathbf{T}_N(\underline{x})\|_2$ and $\| \mathbf{T}_N(\underline{x})^{-1}\|_2$ separately.

Since $|T_k(x)|\leq 1$ for $k\geq 0$ and $x\in[-1,1]$, we have $\|\mathbf{T}_N(\underline{x})\|_2\leq N+1$.
To bound $\|\mathbf{T}_N(\underline{x})^{-1}\|_2$ we note 
that $\mathbf{T}_N(\underline{x}^{cheb})$ is the discrete cosine transform (of type III)~\cite{Strang_99_01}, 
where $\underline{x}^{cheb}$ is the vector of Chebyshev points in~\eqref{eq:ChebyshevPoints}.
Hence, $\mathbf{T}_N(\underline{x}^{cheb})D^{-1/2}$ is an orthogonal matrix with 
$D = {\rm diag}(N+1, (N+1)/2,\ldots,(N+1)/2)$. By the Lagrange interpolation formula~\cite[Sec.~4.1]{Powell_81_01} (applied 
to each entry of $\mathbf{T}_N(\underline{x})$) we have the following matrix decomposition:
\begin{equation}
 \mathbf{T}_N(\underline{x}) = C \mathbf{T}_N(\underline{x}^{cheb}),\qquad C_{ij} = \prod_{k=0,k\neq j}^N \frac{x_i-x_k^{cheb}}{x_j^{cheb}-x_k^{cheb}}.
\label{eq:CauchyDecomposition}
\end{equation} 
Since $\mathbf{T}_N(\underline{x}^{cheb})D^{-1/2}$ is an orthogonal matrix we find that 
\[
 \|\mathbf{T}_N(\underline{x})^{-1}\|_2 = \| D^{-1/2} (\mathbf{T}_N(\underline{x}^{cheb})D^{-1/2})^{-1} C^{-1} \|_2 \leq \sqrt{2}(N+1)^{-1/2}\|C^{-1}\|_2.
\]
We must now bound $\|C^{-1}\|_2$. From~\eqref{eq:CauchyDecomposition} we see that 
$C$ is a generalized Cauchy matrix and hence, there is an explicit formula for 
its inverse given by~\cite[Thm.~1]{Schechter_59_01}  
\begin{equation}
  (C^{-1})_{ij} = \prod_{k=0,k\neq j}^N \frac{x_i^{cheb}-x_k}{x_j-x_k} := \ell_j(x_i^{cheb}), \qquad 0\leq i,j\leq N.
\label{eq:cauchyInverseFormula}
\end{equation} 
By the equivalence of matrix norms, we have $(N+1)^{-1/2}\|C^{-1}\|_\infty\leq\|C^{-1}\|_2\leq(N+1)^{1/2}\|C^{-1}\|_\infty$ 
and from~\eqref{eq:cauchyInverseFormula} we find that
\[
 \|C^{-1}\|_\infty = \sup_{0\leq i\leq N} \sum_{j=0}^N |\ell_j(x_i^{cheb})|\leq \Lambda_N(\underline{x}).
\] 
The upper bound in the statement of the theorem follows by combining the calculated upper bounds 
for $\|\mathbf{T}_N(\underline{x})\|_2$ and $\| \mathbf{T}_N(\underline{x})^{-1}\|_2$. 

For the lower bound, note that there exists a polynomial $p^*$ of degree $N$ such that $|p(x_k)|\leq 1$ and an $x^*\in[-1,1]$ 
such that $|p(x^*)| = \Lambda_N(\underline{x})$. Let $p(x) = \sum_{k=0}^\infty c_k^{cheb}T_k(x)$. Since $|T_k(x)|\leq 1$
and $|p(x^*)| = \Lambda_N(\underline{x})$, there exists an $0\leq k\leq N$ such that $|c_k^{cheb}|\geq \Lambda_N(\underline{x})/(N+1)$. 
Hence, $\|\underline{c}^{cheb}\|_2\geq \Lambda_N(\underline{x})/(N+1)$ and we have
\[
\frac{\Lambda_N(\underline{x})}{N+1}\leq \|\underline{c}^{cheb}\|_2\leq \|\mathbf{T}_N(\underline{x})^{-1}\|_2\|p(\underline{x})\|_2 \leq \sqrt{N+1}\|\mathbf{T}_N(\underline{x})^{-1}\|_2.
\]
The lower bound in the statement of the theorem follows 
from $\|\mathbf{T}_N(\underline{x})\|_2\geq (N+1)^{1/2}\|\mathbf{T}_N(\underline{x})\|_1\geq (N+1)^{3/2}$. 
\end{proof}

Theorem~\ref{thm:illconditioned} explains why  
$\kappa_2\left( \mathbf{T}_N(\underline{x}^{equi})\right)$ 
grows exponentially with $N$ and confirms that one should expect 
severe numerical issues with equally spaced polynomial 
interpolation, in addition to the possibility of Runge's phenomenon.  

It is not the Chebyshev polynomials that should be blamed for the exponential growth of $\kappa_2(\mathbf{T}_N(\underline{x}^{equi}))$
with $N$, but the equally spaced points on $[-1,1]$. 
In a different direction, others have focused on finding $N+1$ points $\underline{x}$ 
such that $\mathbf{T}_N(\underline{x})$ 
is well-conditioned. Reichel and Opfer showed that $\mathbf{T}_N(\underline{x})$ 
is well-conditioned when $\underline{x}$ is a set of points on a certain
Bernstein ellipse~\cite{Reichel_91_01}. Gautschi in~\cite[(27)]{Gautschi_11_01} gives 
an explicit formula for the condition number of $\mathbf{T}_N(\underline{x})$ for any
point set $\underline{x}$ in the Frobenius norm and showed that $\mathbf{T}_N(\underline{x}^{cheb})D^{-1/2}$ 
is the only perfectly conditioned matrix among all so-called Vandermonde-like matrices~\cite{Gautschi_11_01}. A survey of this research area can be 
found here~\cite[Sec.~V]{Gautschi_90_01}. 

%


\section{How good is equally spaced least squares polynomial fitting?}\label{sec:EquallySpacedLeastSquares}
We now turn our attention to solving the least squares problem in~\eqref{eq:leastSquares}, 
where $M<N$. We are interested in the normal equations in~\eqref{eq:normalEquations}
and the condition number of the $(M+1)\times (M+1)$ matrix 
$\mathbf{T}_M(\underline{x}^{equi})^*\mathbf{T}_M(\underline{x}^{equi})$. We 
show that the situation is very different from in Section~\ref{sec:EquallySpacedInterpolation}
if we take $M\leq\constant\sqrt{N}$. In particular, $\kappa_2(\mathbf{T}_M(\underline{x}^{equi})^*\mathbf{T}_M(\underline{x}^{equi}))$
is bounded with $N$ and grows modestly with $M$ if $M\leq \constant\sqrt{N}$.  This 
means that the Chebyshev coefficients for $p_M(x)$ in~\eqref{eq:leastSquares} are not 
sensitive to the perturbations in $\underline{f}+\underline{\eps}$ and can be computed accurately in 
double precision. 

To bound the condition number of 
$\mathbf{T}_M(\underline{x}^{equi})^*\mathbf{T}_M(\underline{x}^{equi})$ we 
can no longer use the matrix decomposition 
in~\eqref{eq:CauchyDecomposition} as that is not applicable when $M<N$. 
Instead, we first consider the normal equations for the Legendre--Vandermonde\footnote{The Legendre--Vandermonde 
matrix in~\eqref{eq:LegendreVandermondeMatrix} is 
the same as the Chebyshev--Vandermonde matrix except the Chebyshev polynomials 
are replaced by Legendre polynomials.} 
matrix 
\begin{equation} 
 \mathbf{P}_M(\underline{x}^{equi}) =\begin{bmatrix} 
 P_0(x_0^{equi}) & \cdots & P_M(x_0^{equi}) \cr 
 \vdots & \ddots & \vdots \cr 
 P_0(x_N^{equi}) & \cdots & P_M(x_N^{equi}) 
\end{bmatrix}\in\mathbb{R}^{(N+1)\times(M+1)}
\label{eq:LegendreVandermondeMatrix}
\end{equation} 
and $P_k(x)$ is the Legendre polynomial of degree $k$~\cite[Sec.~18.3]{Olver_10_01}. Legendre 
polynomials are theoretically convenient for us because they are orthogonal with respect to 
the standard $L^2$ inner-product~\cite[(18.2.1) \& Tab.~18.3.1]{Olver_10_01}, i.e., 
\begin{equation}
 \int_{-1}^1 P_m(x)P_n(x)dx = \begin{cases}\frac{2}{2n+1}, & m=n,\\0, & m\neq n, \end{cases} \qquad 0\leq m,n\leq M.
\label{eq:LegendreOrthogonality}
\end{equation}
Afterwards, in Theorem~\ref{thm:ChebyshevSingularValues} we go back to 
consider $\kappa_2(\mathbf{T}_M(\underline{x}^{equi})^*\mathbf{T}_M(\underline{x}^{equi}))$.

To bound the condition number of $\mathbf{P}_M(\underline{x}^{equi})^*\mathbf{P}_M(\underline{x}^{equi})$
our key insight is to view the $(m,n)$ entry of 
$\tfrac{2}{N}\mathbf{P}_M(\underline{x}^{equi})^*\mathbf{P}_M(\underline{x}^{equi})$
as essentially a trapezium rule approximation of the integral 
in~\eqref{eq:LegendreOrthogonality}. 
Since $\kappa_2(\mathbf{P}_M(\underline{x}^{equi})^*\mathbf{P}_M(\underline{x}^{equi}))$ equals $\sigma_{1}(\mathbf{P}_M(\underline{x}^{equi}))^2$ divided by $\sigma_{M+1}(\mathbf{P}_M(\underline{x}^{equi}))^2$, 
Theorem~\ref{thm:LegendreVandermondeSingularValues} focuses on bounding the squares of the maximum and minimum 
singular values of $\mathbf{P}_M(\underline{x}^{equi})$. 

\begin{theorem}
For any integers $M$ and $N$ satisfying $M\leq \constant\sqrt{N}$ we have
 \[
  \sigma_{1}\!\left(\mathbf{P}_M(\underline{x}^{equi})\right)^2 \leq 2N, \quad 
 \sigma_{M+1}\!\left(\mathbf{P}_M(\underline{x}^{equi})\right)^2 \geq \frac{2N}{5(2M+1)}.
\]
(Tighter but messy bounds can be found in~\eqref{eq:s1} and~\eqref{eq:s2}.) 
\label{thm:LegendreVandermondeSingularValues}
\end{theorem}
\begin{proof} 
If $M = 0$, then $\mathbf{P}_M(\underline{x}^{equi})$ is the $(N+1)\times 1$ 
vector of all ones. Thus, $\sigma_{1}\!\left(\mathbf{P}_M(\underline{x}^{equi})\right)^2 = N$ and 
the bounds above hold. For the remainder of this proof we assume 
that $M\geq 1$ and hence, $N\geq 4$. 
 
From the orthogonality of Legendre polynomials in~\eqref{eq:LegendreOrthogonality} we define
\[
D_{mn} = \frac{N}{2}\int_{-1}^1 P_m(x)P_n(x)dx = \begin{cases}\frac{N}{2n+1}, & m=n,\\0, & m\neq n, \end{cases} \qquad 0\leq m,n\leq M.
\]
The $(N+1)$-point trapezium rule (see~\eqref{eq:trapeziumRule}) provides another expression for $D$, 
\begin{equation}
D = \mathbf{P}_M(\underline{x}^{equi})^* \mathbf{P}_M(\underline{x}^{equi}) - C - \frac{N}{2}E, \qquad C_{mn} = \begin{cases}1,&m+n \text{ is even},\cr 0,&m+n,\text{ is odd},\end{cases}
\label{eq:Cmatrix}
\end{equation} 
where $C$ is the matrix that halves the contributions at the endpoints and 
$E$ is the matrix of trapezium rule errors. By the Euler--Maclaurin error 
formula~\cite[Cor.~3.3]{Javed_14_01} we have, for $0\leq m,n\leq M$, 
\[
E_{mn} = 2\sum_{s=1,s \text{ odd}}^{m+n-1} \frac{((P_m(1)P_n(1))^{(s)} - (P_m(-1)P_n(-1))^{(s)})2^sB_{s+1}}{N^{s+1}(s+1)!},
\]
where $B_{s}$ is the $s$th Bernoulli number and $(P_m(1)P_n(1))^{(s)}$ is the $s$th 
derivative of $P_m(x)P_n(x)$ evaluated at $1$. 
By Markov's brother inequality $|(P_m(1)P_n(1))^{(s)}|\leq 2^s s!(m+n)^{2s}/(2s)!$~\cite[p.~254]{Borwein_95_01}
and since $|B_{s+1}|\leq 4(s+1)!(2\pi)^{-s-1}$~\cite[(24.9.8)]{Olver_10_01} we have 
\[
|E_{mn}| \leq \frac{8}{\pi N} \sum_{s=1,s \text{ odd}}^{m+n-1} \left(\frac{8}{\pi}\right)^s\frac{s!}{(2s)!}\frac{((m+n)/2)^{2s}}{N^{s}} \leq \frac{3(m+n)^2}{\pi N^2},
\]
where in the last inequality we used $((m+n)/2)^2/N\leq M^2/N\leq 1$ and the fact that
$\sum_{s=1,s \text{ odd}}^{m+n-1} (8/\pi)^ss!/(2s)!\leq 3/2$.
Using $\|E\|_2 \leq \|E\|_F$, $(\sum_{m,n=0}^M(m+n)^4)^{1/2}\leq 9M^3/2$, 
and $M\leq \constant \sqrt{N}$, where $\|\cdot\|_F$ denotes the matrix Frobenius 
norm, we obtain
\begin{equation}
 \|E\|_2 \leq \frac{27M^3}{2\pi N^2} \leq \frac{27}{16\pi \sqrt{N}}. 
\label{eq:Ebound}
\end{equation}
By Weyl's inequality on the eigenvalues of perturbed Hermitian 
matrices~\cite{Weyl_1912_01}, we conclude that 
\[
\left|\lambda_{k}\left(\mathbf{P}_M(\underline{x}^{equi})^* \mathbf{P}_M(\underline{x}^{equi})\right) -\lambda_k(D + C)\right|\leq \tfrac{N}{2}\|E\|_2, \qquad 1\leq k\leq M+1,
\]
where $\lambda_k(A)$ denotes the $k$th eigenvalue of the Hermitian matrix $A$. 
By Lemma~\ref{lem:GcirclesDplusC} we have $\lambda_1(D+C)\leq (2N+M+3)/2$ and 
$\lambda_{M+1}(D+C)\geq (N-M^2/2)/(2M+1)$. Since 
$\sigma_{k}(A)^2 = \lambda_{k}(A^*A)$ for any real matrix $A$, we obtain
\begin{equation}
 \sigma_{1}\!\left(\mathbf{P}_M(\underline{x}^{equi})\right)^2 \leq \frac{2N+M+3}{2} + \frac{27\sqrt{N}}{32\pi}
\label{eq:s1}
\end{equation}
and
\begin{equation}
 \sigma_{M+1}\!\left(\mathbf{P}_M(\underline{x}^{equi})\right)^2 \geq \frac{N-M^2/2}{2M+1} - \frac{27\sqrt{N}}{32\pi}.
\label{eq:s2}
\end{equation}
The statement follows since for $M\leq \constant\sqrt{N}$ and $N\geq 4$ we have $ (2N+M+3)/2 + (27\sqrt{N})/(32\pi)\leq 2N$
and $(N-M^2/2)/(2M+1) - (27\sqrt{N})/(32\pi)\geq 2N/(5(2M+1))$.
\end{proof}

Theorem~\ref{thm:LegendreVandermondeSingularValues} shows that if $M\leq \constant\sqrt{N}$, then
\[
\kappa_2(\mathbf{P}_M(\underline{x}^{equi})^*\mathbf{P}_M(\underline{x}^{equi})) \leq 5(2M+1).
\]
This means that when $M\leq \constant\sqrt{N}$ 
we can solve for the Legendre coefficients of $p_M(x)$ in~\eqref{eq:leastSquares} via 
the normal equations, 
\begin{equation}
 \mathbf{P}_M(\underline{x}^{equi})^*\mathbf{P}_M(\underline{x}^{equi})\underline{c}^{leg} = \mathbf{P}_M(\underline{x}^{equi})^*\left(\underline{f}+\underline{\eps}\right),
\label{eq:LegendreNormalEquations}
\end{equation}
without severe ill-conditioning. Here, $\underline{c}^{leg}$ is the vector of coefficients so that
\[
 p_M(x) = \sum_{k=0}^{M} c_k^{leg} P_k(x). 
\]
Hence, the least squares problem 
in~\eqref{eq:leastSquares} is a practical way to construct a polynomial 
approximant of a function from equally spaced samples. 

The bounds in Theorem~\ref{thm:LegendreVandermondeSingularValues} are essentially tight. In Figure~\ref{fig:LegendreVandermondeEigenvalues} we compare the bounds 
in~\eqref{eq:s1} and~\eqref{eq:s2} to 
computed values of the square of maximum and minimum singular values of $\mathbf{P}_M(\underline{x}^{equi})$ when $M=\lfloor \constant\sqrt{N}\rfloor$.
The jagged nature of the bound in Figure~\ref{fig:LegendreVandermondeEigenvalues} is due to the floor function in the formula
for $M$ to ensure it is an integer. This causes jumps in the bounds at each square number. 

\begin{figure}
\centering 
\begin{overpic}[width=.6\textwidth]{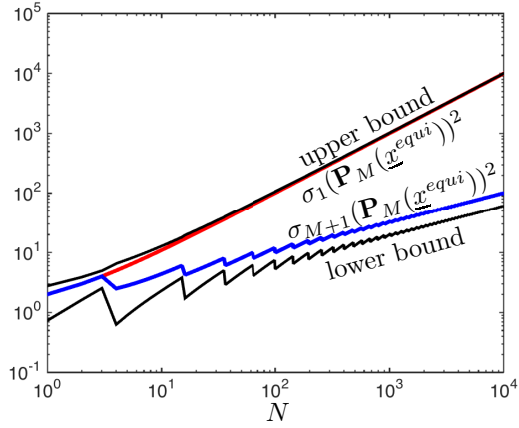}
 \put(50,0) {$N$}
 \put(53,31) {\rotatebox{15}{$\sigma_{M+1}(\mathbf{P}_M(\underline{x}^{equi}))^2$}}
 \put(55,37) {\rotatebox{25}{$\sigma_{1}(\mathbf{P}_M(\underline{x}^{equi}))^2$}}
 \put(55,43) {\rotatebox{27}{upper bound}}
 \put(60,24) {\rotatebox{15}{lower bound}}
\end{overpic}
 \caption{Illustration of the bounds on the squares of the maximum and minimum 
 singular values of $\mathbf{P}_M(\underline{x}^{equi})$ as found in~\eqref{eq:s1} and~\eqref{eq:s2} when $M=\lfloor \constant\sqrt{N} \rfloor$. 
 The statement of Theorem~\ref{thm:LegendreVandermondeSingularValues} provides simplified and slightly weaker bounds.}
\label{fig:LegendreVandermondeEigenvalues}
\end{figure}


One can also use Theorem~\ref{thm:LegendreVandermondeSingularValues} to safely compute the Chebyshev 
coefficients of the polynomial $p_M(x)$ in~\eqref{eq:leastSquares} too. 
Let $S$ be the $(M+1)\times (M+1)$ change of basis matrix that takes Legendre 
coefficients to Chebyshev coefficients. The entries of $S$ have an explicit formula given by~\cite[(2.18)]{Alpert_91_01}
\begin{equation}
 S_{ij} = \begin{cases}\frac{1}{\pi}\Psi\left(\frac{j}{2}\right)^2, & 0=i\leq j\leq M, \text{ }j \text{ even},\cr \frac{2}{\pi}\Psi\left(\frac{j-i}{2}\right)\Psi\left(\frac{j+i}{2}\right),& 0<i\leq j\leq M, \text{ } i+j \text{ even}, \cr 0, & \text{otherwise},\end{cases}
\label{eq:S}
\end{equation} 
where $\Psi(i) = \Gamma(i + 1/2)/\Gamma(i+1)$ and $\Gamma(x)$ is the Gamma function. Theorem~\ref{thm:LegendreVandermondeSingularValues} shows that 
when $M\leq \constant\sqrt{N}$ the Legendre coefficients $\underline{c}^{leg}$ of $p_M(x)$ in~\eqref{eq:leastSquares} can 
be computed accurately via the normal equations in~\eqref{eq:LegendreNormalEquations}.
Afterwards, the Legendre coefficients, $\underline{c}^{leg}$, for $p_M(x)$ can be converted into Chebyshev coefficients, $\underline{c}^{cheb}$,
for $p_M(x)$ by a matrix-vector product, i.e., 
$\underline{c}^{cheb} = S\underline{c}^{leg}$. For fast algorithms to compute the matrix-vector product
$S\underline{c}^{leg}$, see~\cite{Alpert_91_01,Hale_14_01}. 

We rarely compute the Chebyshev coefficients of $p_M(x)$ via the Legendre coefficients from~\eqref{eq:LegendreNormalEquations}. 
Instead, we directly compute the Chebyshev coefficients via the normal 
equations in~\eqref{eq:normalEquations} because we have a fast direct algorithm (see Section~\ref{sec:LeastSquaresAlgorithm}). 
Here, is the bound we obtain on  
$\kappa_2(\mathbf{T}_M(\underline{x}^{equi})^*\mathbf{T}_M(\underline{x}^{equi}))$. 

\begin{theorem} 
For any integers $M$ and $N$ satisfying $M\leq \constant\sqrt{N}$ we have
\[
 \sigma_1\!\left(\mathbf{T}_M(\underline{x}^{equi})\right)^2 \leq 3N,\qquad  \sigma_{M+1}\!\left(\mathbf{T}_M(\underline{x}^{equi})\right)^2 \geq \frac{1}{25} \sigma_{M+1}\!\left(\mathbf{P}_M(\underline{x}^{equi})\right)^2.
\]
\label{thm:ChebyshevSingularValues}
\end{theorem}
\begin{proof}
By the trapezium rule (see~\eqref{eq:integralTmTn}) we have 
\[
 \mathbf{T}_M(\underline{x}^{equi})^*\mathbf{T}_M(\underline{x}^{equi}) = F + C + \tfrac{N}{2}\tilde{E},
\]
where $\tilde{E}$ is the matrix of trapezium errors, $C$ is given in~\eqref{eq:Cmatrix}, and $F$ is given by
\begin{equation}
  F_{mn} = \frac{N}{2}\int_{-1}^1 T_m(x)T_n(x)dx = \frac{N}{2}\left[\frac{1}{1-(m+n)^2}+\frac{1}{1-(m-n)^2}\right].
\label{eq:F}
\end{equation} 
By the same argument as in Theorem~\ref{thm:LegendreVandermondeSingularValues} we have $\|\tilde{E}\|_2\leq 27/(16\pi/\sqrt{N})$, see~\eqref{eq:Ebound}.
Also by Lemma~\ref{lem:GcirclesDplusC2} we have $\lambda_1(F+C)\leq (4N+M+1)/2$ and hence, using Weyl's 
inequality on the eigenvalues of perturbed Hermitian matrices~\cite{Weyl_1912_01} we obtain
\[
 \sigma_1\!\left(\mathbf{T}_M(\underline{x}^{equi})\right)^2\leq \frac{4N+M+1}{2} + \frac{27\sqrt{N}}{32\pi}\leq 3N,
\]
where the last inequality holds since $M\leq \constant \sqrt{N}$. 

Next, by the definition of the matrix $S$ in~\eqref{eq:S} we have 
$\mathbf{T}_M(\underline{x}^{equi})S = \mathbf{P}_M(\underline{x}^{equi})$. Hence,
\[
\sigma_{M+1}\left(\mathbf{T}_M(\underline{x}^{equi})\right)^2 \geq \|S\|_2^{-2}\sigma_{M+1}\left(\mathbf{P}_M(\underline{x}^{equi})\right)^2.
\]
The lower bound on $\sigma_{M+1}\left(\mathbf{T}_M(\underline{x}^{equi})\right)$ 
follows from Lemma~\ref{lem:BS}, which proves $\|S\|_2\leq 5$. 
\end{proof}

Theorem~\ref{thm:ChebyshevSingularValues} bounds the condition number of $\mathbf{T}_M(\underline{x}^{equi})^*\mathbf{T}_M(\underline{x}^{equi})$.
If $M\leq \constant\sqrt{N}$, then
\begin{equation}
\kappa_2(\mathbf{T}_M(\underline{x}^{equi})^*\mathbf{T}_M(\underline{x}^{equi})) \leq \frac{375}{2}(2M+1).
\label{eq:ChebyshevConditionNumber}
\end{equation} 
Therefore, the Chebyshev coefficients, $\underline{c}^{cheb}$, of $p_M(x)$ in~\eqref{eq:leastSquares} can be computed accurately 
via the normal equations in~\eqref{eq:normalEquations}. 

The lower bound on the minimum singular value of $\mathbf{T}_M(\underline{x}^{equi})$ 
shows that the solution vector $\underline{c}^{cheb}$ is not sensitive to small perturbations 
in the function samples and is the key result for sections~\ref{sec:StableLeastSquares} and~\ref{sec:Extrapolation}. 

The $M\leq \constant\sqrt{N}$ assumption in Theorem~\ref{thm:ChebyshevSingularValues} can in practice be slightly violated
without consequence, for example, $M\leq 2\sqrt{N}$ gives the same qualitative behavior. We can even improve 
the restriction in Theorem~\ref{thm:ChebyshevSingularValues} 
to $M\leq 0.95\sqrt{N}$ and use the same argument to show that $\kappa_2(\mathbf{T}_M(\underline{x}^{equi})^*\mathbf{T}_M(\underline{x}^{equi}))$ 
grows linearly with $M$. Unfortunately, the derived bounds are so awkward to write down that they are
not worthwhile in a paper of this nature. We certainly do not pretend that the constants
in Theorem~\ref{thm:LegendreVandermondeSingularValues} and Theorem~\ref{thm:ChebyshevSingularValues} are tight, 
though we are pleased that the bounds are explicit. 

During the final stages of writing this paper we were made aware 
of~\cite[Thm.~5.1]{Adcock_15_02}, which as a special case also gives a similar, but non-explicit, bound 
as in Theorem~\ref{thm:ChebyshevSingularValues}. Since we are using very specific techniques,
the bounds in Theorem~\ref{thm:ChebyshevSingularValues} have explicit constants. 

\section{Approximation power of least squares polynomial fitting}\label{sec:ApproximationPower}
In this section, we derive results to understand how well $p_M$ approximates $f$ on $[-1,1]$ under the 
assumption that $\eps = 0$, i.e., the function samples are not perturbed. 

The following theorem allows for any $M\leq N$, though afterwards we restrict $M\leq\constant\sqrt{N}$
so that $\sigma_{M+1}(\mathbf{T}_M(\underline{x}^{equi}))$ can be bounded from below using Theorem~\ref{thm:ChebyshevSingularValues}. 
\begin{theorem}
Let $M$ and $N$ be integers satisfying $M\leq N$ and $\eps = 0$. Let $f$ be an analytic function with Bernstein parameter $\rho>1$ 
and $\underline{c}^{cheb}$ be the vector of Chebyshev coefficients of the degree $M$ 
polynomial $p_M$ in~\eqref{eq:leastSquares}. Then, we have
\[
|c_k^{cheb}| \leq 2Q\left[\rho^{-k} + \frac{(N+1)^{1/2}}{\sigma_{M+1}(\mathbf{T}_M(\underline{x}^{equi}))}\frac{\rho^{-M}}{\rho-1}\right],\qquad 0\leq k\leq M, 
\]
and 
\[
 \|f - p_M\|_\infty \leq 2Q\left[1 + \frac{(M+1)(N+1)^{1/2}}{\sigma_{M+1}(\mathbf{T}_M(\underline{x}^{equi}))}\right]\frac{\rho^{-M}}{\rho-1},
\]
where $\|f - p_M\|_\infty=\sup_{x\in[-1,1]} |f(x)-p_M(x)|$. 
\label{thm:AnalyticLeastSquares}
\end{theorem}
\begin{proof} 
Let $f_M$ be the polynomial of degree $M$ constructed by truncating the 
Chebyshev expansion for $f$ after $M+1$ terms, see Proposition~\ref{prop:AnalyticConvergence}. Then, 
\[
 f(\underline{x}^{equi}) = f_M(\underline{x}^{equi}) + f(\underline{x}^{equi})- f_M(\underline{x}^{equi}) = \mathbf{T}_M(\underline{x}^{equi})\underline{a}_M + (f- f_M)(\underline{x}^{equi}),
\]
where $\underline{a}_M$ is the vector of the first $M+1$ Chebyshev coefficients for $f$. The vector $\underline{c}^{cheb}$ 
satisfies the normal equations, $\mathbf{T}_M(\underline{x}^{equi})^*\mathbf{T}_M(\underline{x}^{equi})\underline{c}^{cheb}=\mathbf{T}_M(\underline{x}^{equi})^*f(\underline{x}^{equi})$, 
and since $f_M(\underline{x}^{equi})=\mathbf{T}_M(\underline{x}^{equi})\underline{a}_M$ we have
\[
\mathbf{T}_M(\underline{x}^{equi})^*\mathbf{T}_M(\underline{x}^{equi}) (\underline{c}^{cheb}-\underline{a}_M) = \mathbf{T}_M(\underline{x}^{equi})^*(f - f_M)(\underline{x}^{equi}).
\]
Noting that $\|(A^*A)^{-1}A^*\|_2 = 1/\sigma_{\min}(A)$ for any matrix $A$, we have the following bound:
\begin{equation}
\begin{aligned} 
\|\underline{c}^{cheb}-\underline{a}_M\|_\infty 
&\leq \|(\mathbf{T}_M(\underline{x}^{equi})^*\mathbf{T}_M(\underline{x}^{equi}) )^{-1}\mathbf{T}_M(\underline{x}^{equi})^*\|_\infty\|(f - f_M)(\underline{x}^{equi})\|_\infty\\
&\leq 2Q \frac{(N+1)^{1/2}}{\sigma_{M+1}(\mathbf{T}_M(\underline{x}^{equi}))}\frac{\rho^{-M}}{\rho-1},
\end{aligned} 
\label{eq:UsefulInequality}
\end{equation}
where in the last inequality we used 
$\|(f-f_M)(\underline{x}^{equi})\|_\infty\leq 2Q\rho^{-M}/(\rho-1)$ (see Proposition~\ref{prop:AnalyticConvergence}) and 
$\|A\|_\infty\leq \sqrt{N+1}\|A\|_2$ for matrices of size $(M+1)\times (N+1)$. The
bound on $|c_k^{cheb}|$ follows since $|c_k^{cheb}|\leq |a_k| + \|\underline{c}^{cheb}-\underline{a}_M\|_\infty$ and 
$|a_k|\leq 2Q\rho^{-k}$ for $k\geq 0$ (see Proposition~\ref{prop:AnalyticConvergence}).

For a bound on $\|f-p_M\|_\infty$, note that $T_k(x)\leq 1$ for $k\geq 0$ and $x\in[-1,1]$. Hence, 
\[
 \|f-p_M\|_\infty \leq (M+1)\|\underline{c}^{cheb}-\underline{a}_M\|_\infty + \sum_{k=M+1}^\infty\!\! |a_k| \leq 2Q\!\left[1 + \frac{(M+1)(N+1)^{1/2}}{\sigma_{M+1}(\mathbf{T}_M(\underline{x}^{equi}))}\right]\!\frac{\rho^{-M}}{\rho-1},
\]
where we again used $|a_k|\leq 2Q\rho^{-k}$ for $k\geq 0$. 
\end{proof}

When $M \leq \constant\sqrt{N}$ we can use Theorem~\ref{thm:AnalyticLeastSquares} together with the lower 
bound on $\sigma_{M+1}(\mathbf{T}_M(\underline{x}^{equi}))$ from Theorem~\ref{thm:ChebyshevSingularValues} to conclude that 
\begin{equation}
 \|f - p_M\|_\infty \leq 2Q\left[1 + 10\sqrt{5}(M+1)^{3/2}\right]\frac{\rho^{-M}}{\rho-1}.
\label{eq:ApproximationBound}
\end{equation} 
Thus, with respect to $M$, the least squares polynomial fit $p_M$ converges geometrically to $f$ with order $\rho$. 
Along with the bound on the condition number 
of the normal equations in~\eqref{eq:ChebyshevConditionNumber}, it confirms that least squares polynomial approximation 
is a practical tool for approximating analytic functions given equally spaced samples. 

It is common to refer to~\eqref{eq:ApproximationBound} as a subexponential convergence rate 
because one needs to take $\mathcal{O}(N)$ equally spaced samples to realize an approximation error of 
$\mathcal{O}(\rho^{-\sqrt{N}})$. We now use the noisy bounds to consider the case when the 
function samples are perturbed, i.e., $\eps>0$. 
%

\section{Least squares polynomial fitting is robust to noisy samples}\label{sec:StableLeastSquares}
Polynomial interpolation at equally spaced points is sensitive to noisy function samples
samples, and this is a considerable drawback. 
In contrast, when there is sufficient oversampling, i.e., $M\leq \constant\sqrt{N}$, least squares 
polynomial fits are robust to perturbed function samples.
In this section we consider two cases:  
The vector of function samples $f(\underline{x}^{equi})$ is 
perturbed by either a vector of independent Gaussian random variables with mean $0$ and known variance $s^2$ or 
a vector of deterministic errors given by $\underline{\eps}$ with known maximum 
amplitude $\|\underline{\eps}\|_\infty$. 

\subsection{Least squares polynomial fitting with Gaussian noise}\label{sec:GaussianError}
First suppose that the samples are given by 
$f(\underline{x}^{equi}) + \underline{\eps}$, where 
$\underline{\eps} = (\eps_0,\ldots,\eps_N)^T$ is a vector with 
entries that are independent
Gaussian random variables with mean $0$ and variation $s^2$, i.e.,
$\eps_k\sim\mathcal{N}(0,s^2)$ for $0\leq k\leq N$. 
We refer to the standard deviation of the noise, $s$, as the {\em noise level}. 

Thus, we seek the solution to 
the perturbed normal equations,
\begin{equation}
 \mathbf{T}_M(\underline{x}^{equi})^*\mathbf{T}_M(\underline{x}^{equi})\underline{c}^{cheb} = 
\mathbf{T}_M(\underline{x}^{equi})^*\left(\underline{f}+\underline{\eps}\right).
\label{eq:normalequationsUnderNoise}
\end{equation} 
The vector of Chebyshev coefficients $\underline{c}^{cheb}$ for the least 
squares fit $p_M(x)$ are now a vector of random variables.  It is easy to see that the 
expectation of the vector $\underline{c}^{cheb}$ is given by
\[
 \mathbb{E}\left[ \underline{c}^{cheb} \right] = (\mathbf{T}_M(\underline{x}^{equi})^*\mathbf{T}_M(\underline{x}^{equi}))^{-1}\mathbf{T}_M(\underline{x}^{equi})^*\underline{f},
\]
which verifies that the expectation of the coefficients is the same as in the noiseless case. 
To get a bound on the expectation of the approximation error $\|f-p_M\|_\infty$ we need to 
bound the variance of $\underline{c}^{cheb}$. Here, is one 
such bound that we state for any $M\leq N$, though afterwards we restrict ourselves to 
$M\leq \constant\sqrt{N}$.
\begin{lemma} 
Let $M$ and $N$ be integers satisfying $M\leq N$ 
and let $\underline{\eps}\in\mathbb{R}^{(N+1)\times 1}$ be a vector with entries that are realizations from 
independent and identically distributed Gaussian random variables with mean $0$ and variance $s^2$. Then, 
for the vector $\underline{c}^{cheb}$ satisfying~\eqref{eq:normalequationsUnderNoise} we have
\[
 \mathbb{E}\left[\left\|\underline{c}^{cheb}-\mathbb{E}\left[\underline{c}^{cheb}\right]\right\|_2^2\right] \leq \frac{(M+1)s^2}{\sigma_{M+1}(\mathbf{T}_M(\underline{x}^{equi}))^2},
\] 
where $\mathbb{E}[X]$ denotes the expectation of the random variable $X$.
\label{lem:varianceNoise}
\end{lemma}
\begin{proof} 
Let $A = \mathbf{T}_M(\underline{x}^{equi})$ and 
note that 
\[
 \mathbb{E}\left[\left\|\underline{c}^{cheb}-\mathbb{E}\left[\underline{c}^{cheb}\right]\right\|_2^2\right] = \mathbb{E}\!\left[\left\|(A^*A)^{-1}A^*\underline{\eps}\right\|_2^2\right].
\]
Let $P = A(A^*A)^{-1}A^*\in\mathbb{C}^{(N+1)\times (N+1)}$ be the orthogonal projection of $\mathbb{C}^{N+1}$ onto the range of $A$. Since 
$(A^*A)^{-1}A^* = (A^*A)^{-1}A^*P$,
$\|(A^*A)^{-1}A^*\|_2 = \sigma_{M+1}(A)^{-1}$, and\footnote{Let $A = QR$ be the reduced QR factorization of $A$, where $Q=\left[\underline{q}_0\,|\,\cdots\,|\,\underline{q}_{M}\right]$. Then, $\|P\underline{\eps}\|_2^2 = \underline{\eps}^*P^*P\underline{\eps} = \underline{\eps}^*QQ^*\underline{\eps}=\sum_{k=0}^M|\underline{q}_k^*\underline{\eps}|^2$. Since $\mathbb{E}[|\underline{q}_k^*\underline{\eps}|^2]= s^2$ we have $\mathbb{E}[\|P\underline{\eps}\|_2^2]\leq (M+1)s^2$.} $\mathbb{E}[\|P\underline{\eps}\|_2^2]\leq (M+1)s^2$, we have 
\[
\mathbb{E}\!\left[\|(A^*A)^{-1}A^*\underline{\eps}\|_2^2\right]\! =\! \mathbb{E}\!\left[\|(A^*A)^{-1}A^*P\underline{\eps}\|_2^2\right]\!\leq\! \|(A^*A)^{-1}A^*\|_2^2\,\mathbb{E}\!\left[\|P\underline{\eps}\|_2^2\right]\!\leq\! \frac{(M+1)s^2}{\sigma_{M+1}(A)^2}, 
\]
as required.
\end{proof}

Lemma~\ref{lem:varianceNoise} shows that the sum of the variances of the entries of 
$\underline{c}^{cheb}$ is comparable to the sum of the variances of $\underline{\eps}$, 
provided that $\sigma_{M+1}(\mathbf{T}_M(\underline{x}^{equi}))$ is not too small. Thus, 
if $\sigma_{M+1}(\mathbf{T}_M(\underline{x}^{equi}))$ is sufficiently large, then we expect
$p_M(x)$ to be stable under small perturbations of the function samples. We show this by bounding 
the expected maximum uniform error between $f$ and $p_M$.

\begin{corollary}
Suppose the assumptions of Lemma~\ref{lem:varianceNoise} hold, $f$ is 
an analytic function with Bernstein parameter $\rho>1$, and $p_M$ is the least squares polynomial 
fit of degree $M$ in~\eqref{eq:leastSquares}. Then, 
 \begin{equation}
 \mathbb{E}\left[\left\|f - p_M\right\|_\infty\right] \leq \frac{(M+1)^{3/2}s}{\sigma_{M+1}(\mathbf{T}_M(\underline{x}^{equi}))} + 2Q\left[1 + \frac{(M+1)(N+1)^{1/2}}{\sigma_{M+1}(\mathbf{T}_M(\underline{x}^{equi}))}\right]\frac{\rho^{-M}}{\rho-1}.
\label{eq:NoiseApproxPower}
\end{equation}
Moreover, when $M\leq \constant\sqrt{N}$ we have
\begin{equation}
  \mathbb{E}\left[\left\|f - p_M\right\|_\infty\right] \leq \frac{5\sqrt{5}(M+1)^{2}s}{N^{1/2}} + 2Q\left[1 + 10\sqrt{5}(M+1)^{3/2}\right]\frac{\rho^{-M}}{\rho-1},
\label{eq:NoiseApproximationPower}
\end{equation}
where $\|f - p_M\|_\infty=\sup_{x\in[-1,1]} |f(x)-p_M(x)|$.
\label{cor:ApproximationPowerWithNoise}
\end{corollary}
\begin{proof} 
 The same reasoning as in Theorem~\ref{thm:AnalyticLeastSquares}, but with an extra term allowing for the noisy samples, gives the bound   
 \[
 \begin{aligned} 
  \mathbb{E}\left[\left|f(x) - p_M(x)\right|\right] \leq 2Q&\left[1 + \frac{(M+1)(N+1)^{1/2}}{\sigma_{M+1}(\mathbf{T}_M(\underline{x}^{equi}))}\right]\frac{\rho^{-M}}{\rho-1} \\
  &\qquad + \mathbb{E}\left[\left|\sum_{k=0}^M (\underline{c}^{cheb} - \mathbb{E}[\underline{c}^{cheb}])_kT_k(x)\right|\right].
 \end{aligned}
 \]
 Since $|T_k(x)|\leq 1$ and $\mathbb{E}[X]^2\leq \mathbb{E}[X^2]$, this extra term can be bounded as follows: 
 \[
 \begin{aligned}
 \mathbb{E}\left[\left|\sum_{k=0}^M (\underline{c}^{cheb} - \mathbb{E}[\underline{c}^{cheb}])_kT_k(x)\right|\right]^2 &\leq 
 \mathbb{E}\left[\left\|\underline{c}^{cheb} - \mathbb{E}[\underline{c}^{cheb}]\right\|_1^2\right] \\ 
 & \leq (M+1)^2\mathbb{E}\left[\left\|\underline{c}^{cheb} - \mathbb{E}[\underline{c}^{cheb}]\right\|_2^2\right],
 \end{aligned}
 \]
where Lemma~\ref{lem:varianceNoise} can now be employed. For the second
statement, substitute the bound derived in~\eqref{eq:ApproximationBound} into~\eqref{eq:NoiseApproxPower}.
\end{proof}

Therefore, when $M\leq \constant\sqrt{N}$ the least squares polynomial 
fit $p_M(x)$ in~\eqref{eq:leastSquares} is robust to noisy equally spaced
samples of $f$. 
On closer inspection of the bound in~\eqref{eq:NoiseApproximationPower}, we find 
that $\|f-p_M\|_\infty$ decays geometrically with order $\rho$, until it plateaus 
at roughly $5\sqrt{5}(M+1)^2/\sqrt{N} s$.  Since $M\approx\constant\sqrt{N}$ the 
plateau is proportional to the noise level $s$, even as $N\rightarrow \infty$.

One interesting regime  is to keep $M$ fixed and to increase 
the number of samples $N+1$. We see that $\|f-p_M\|_\infty$ is about $\mathcal{O}(s/\sqrt{N})$ in size. 
Intuitively, this makes sense 
because one could imagine averaging nearby samples onto a coarser equally spaced 
grid and using those averaged samples instead. Since the variance of an average of independent random variables scales like
the reciprocal of the number in the average, we expect $\|f-p_M\|_\infty$ to plateau at about 
$\mathcal{O}(s/\sqrt{N})$.  Figure~\ref{fig:NoisySamples} shows a related phenomenon on the plateau 
of the Chebyshev coefficients of the least squares polynomial fit $p_M(x)$ to $f(x) = 1/(1+25(x-1/100)^2)$ on $[-1,1]$.
If the number of samples is increased by a factor of $100$, then the plateau of the Chebyshev 
coefficients drops by a factor of $10$, which confirms the $s/\sqrt{N}$ behavior.

\begin{figure} 
\centering
\begin{overpic}[width=.6\textwidth]{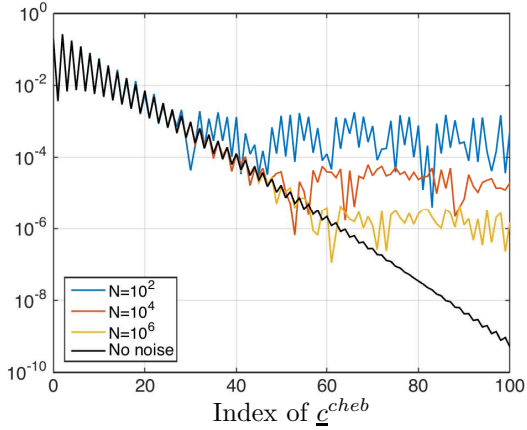}
\put(40,0) {Index of $\underline{c}^{cheb}$}
\end{overpic} 
\caption{The Chebyshev coefficients of $p_M(x)$ in~\eqref{eq:leastSquares} when $M=100$, $f(x) = 1/(1+25(x-1/100)^2)$, and 
the function samples are perturbed by white noise with a standard deviation of $10^{-3}$. The shift of $1/100$ in the definition of $f$ 
is to prevent the function from being even, simplifying the plot. When the number of equally spaced samples is increased by a factor of $100$, the 
plateau in the tail of the coefficients drops by a factor of $10$ (see Corollary~\ref{cor:ApproximationPowerWithNoise}).}
\label{fig:NoisySamples}
\end{figure} 

\subsection{Least squares polynomial fitting with deterministic perturbations}\label{sec:deterministicError}
Now suppose that $f(\underline{x}^{equi})$ is polluted with deterministic error
such as $f(\underline{x}^{equi}) + \underline{\eps}$, where 
$\underline{\eps} = (\eps_0,\ldots,\eps_N)^T$ is a vector such that 
$\|\underline{\eps}\|_\infty=\eps<\infty$. 
We wish to solve  
\[
 \mathbf{T}_M(\underline{x}^{equi})^*\mathbf{T}_M(\underline{x}^{equi})\underline{c}^{cheb} = 
\mathbf{T}_M(\underline{x}^{equi})^*\left(\underline{f}+\underline{\eps}\right)
\] 
and understand the quality of the resulting least squares polynomial fit $p_M$. This is relatively 
easy to do given the proof of Theorem~\ref{thm:AnalyticLeastSquares} so we state it as a corollary. 

\begin{corollary} 
Suppose that the assumptions in Theorem~\ref{thm:AnalyticLeastSquares} are satisfied,
and that the values of $\underline{f}$ are perturbed by a vector $\underline{\eps}$, where 
$\|\underline{\eps}\|_\infty = \eps <\infty$. Then, 
\[
 \|f - p_M\|_\infty \leq 2Q\left[1 + \frac{(M+1)(N+1)^{1/2}}{\sigma_{M+1}(\mathbf{T}_M(\underline{x}^{equi}))}\right]\frac{\rho^{-M}}{\rho-1} + \frac{(M+1)(N+1)^{1/2}}{\sigma_{M+1}(\mathbf{T}_M(\underline{x}^{equi}))}\eps.
\]
\label{cor:deterministicError}
\end{corollary}
\begin{proof} 
 The same proof as Theorem~\ref{thm:AnalyticLeastSquares} except with an additional term that is 
 easy to bound due to the vector $\underline{\eps}$.
\end{proof}

By taking $M\leq \constant\sqrt{N}$ and noting that $\sigma_{M+1}(\mathbf{T}_M(\underline{x}^{equi}))^2\geq 2N/(125(2M+1))$ 
we see that $p_M$ does not converge to $f$ as $N\rightarrow\infty$ with deterministic error, though it plateaus 
at around $\mathcal{O}(\eps)$.  If $M$ is fixed and $N\rightarrow\infty$, 
then Corollary~\ref{cor:deterministicError} shows that $\|f - p_M\|_\infty$ remains bounded.
This is to be expected because in this situation one cannot average a dense set of function samples onto a coarse grid and reduce the 
uncertainty in the sampled values. 

\section{Stable extrapolation with least squares polynomial fits}\label{sec:Extrapolation}
Without perturbed function samples a least squares polynomial fit from equally spaced samples can be 
used to extrapolate outside of $[-1,1]$, by a distance that depends on the analyticity of the 
sampled function. However, in practice polynomial extrapolation is sensitive to perturbed samples or 
roundoff errors in floating point arithmetic. 

In sections~\ref{sec:ExtrapolationWithNoise} and~\ref{sec:ExtrapolationWithDeterministicNoise} we go further and 
show that there are two interesting regimes: (1) if the noise is modeled by independent Gaussian random variables and there is 
exponential oversampling, i.e., $N=e^{aM}$, then 
one can stably extrapolate to $x\in I_\rho = [1,(\rho+1/\rho)/2)$, and (2)
if the noise in the function samples is deterministic, then  
there is a degree $M^*$ that (nearly) minimizes $\sup_{x\in I_\rho} |f(x)-p_M(x)|$.  If $M^*<N/2$, then 
the minimum extrapolation error is a $x$-dependent fractional power of $\eps$. 

\subsection{Extrapolation without noise}\label{sec:ExtrapolationWithoutNoise}
Without noise in the function samples, it turns out that one can extrapolate by any $x$ satisfying 
$1\leq x<(\rho+\rho^{-1})/2$. 
In fact, one cannot expect to extrapolation any further than $(\rho+\rho^{-1})/2$ with a polynomial approximant 
because $f$ is only assumed to be bounded and analytic in an ellipse 
that intercepts the $x$-axis at $(\rho+\rho^{-1})/2$.  

\begin{theorem}
Suppose that the assumptions in Theorem~\ref{thm:AnalyticLeastSquares} hold. 
Then, for any  $1 < x < (\rho+\rho^{-1})/2$ we have 
 \[
 |f(x) - p_M(x)|\leq 2Q\left[\frac{(N+1)^{1/2}(M+1)}{\sigma_{M+1}(\mathbf{T}_M(\underline{x}^{equi}))(\rho-1)} + \frac{r}{1-r}\right]r^M,
\]
where $r = (x + \sqrt{x^2-1})/\rho<1$.  In other words, it is possible to use 
$p_M(x)$ to extrapolate outside of $[-1,1]$ by a distance
determined by the analyticity of $f$. 
\label{thm:extrapolation} 
\end{theorem}
\begin{proof} 
Since $|T_k(x)|\leq (x + \sqrt{x^2-1})^k$ for $x>1$ we have, by Theorem~\ref{thm:AnalyticLeastSquares} and Proposition~\ref{prop:AnalyticConvergence},
\[
\begin{aligned} 
 |f(x) - p_M(x)| &\leq \sum_{k=0}^M\! |a_k - c_k^{cheb}||T_k(x)| + \sum_{k=M+1}^\infty |a_k||T_k(x)| \\
 & \leq \|\underline{a}_M - \underline{c}^{cheb}\|_\infty \sum_{k=0}^M|T_k(x)| + 2Q \sum_{k=M+1}^\infty\rho^{-k}|T_k(x)|\\
 & \leq 2Q\left[\frac{(N+1)^{1/2}}{\sigma_{M+1}(\mathbf{T}_M(\underline{x}^{equi}))}\frac{\rho^{-M}}{\rho-1}\sum_{k=0}^M \rho^kr^k + r^{M+1}\sum_{k=0}^\infty r^k\right],
\end{aligned} 
\]
where the last inequality used~\eqref{eq:UsefulInequality} and $r = (x + \sqrt{x^2-1})/\rho$. Since $1<x<(\rho+\rho^{-1})/2$ we have 
$r<1$ and by the sum of a geometric series we conclude that 
\[ 
 |f(x) - p_M(x)|\leq 2Q\left[\frac{(N+1)^{1/2}(M+1)}{\sigma_{M+1}(\mathbf{T}_M(\underline{x}^{equi}))(\rho-1)} + \frac{r}{1-r}\right]r^M,
\] 
where we used the inequality $\sum_{k=0}^M \rho^kr^k\leq (M+1)\rho^{M}r^M$.
\end{proof}

Figure~\ref{fig:StableExtrapolation} verifies Theorem~\ref{thm:extrapolation} for 
$f(x) = 1/(1+x^2)$ and $g(x) = 1/(1+2x^2)$. Let $p$ and $q$ be the least squares polynomial fits to $f$ and $g$ of degree $M$ constructed
by $N+1$ equally spaced samples with $M\leq \constant \sqrt{N}$. 
Since $\rho = 2.42$ is the Bernstein parameter for $f$ and $\rho=4.24$ for $g$, Theorem~\ref{thm:extrapolation} predicts 
that the least squares error $|f(x)-p_M(x)|$ and $|g(x)-q(x)|$ geometrically decays to zero as $M\rightarrow\infty$ for $1<x<\sqrt{2}$ and 
$1<x<1.23$, respectively. This is observed in Figure~\ref{fig:StableExtrapolation}.
 \begin{figure} 
\centering
\begin{minipage}{.49\textwidth} 
 \begin{overpic}[width=\textwidth]{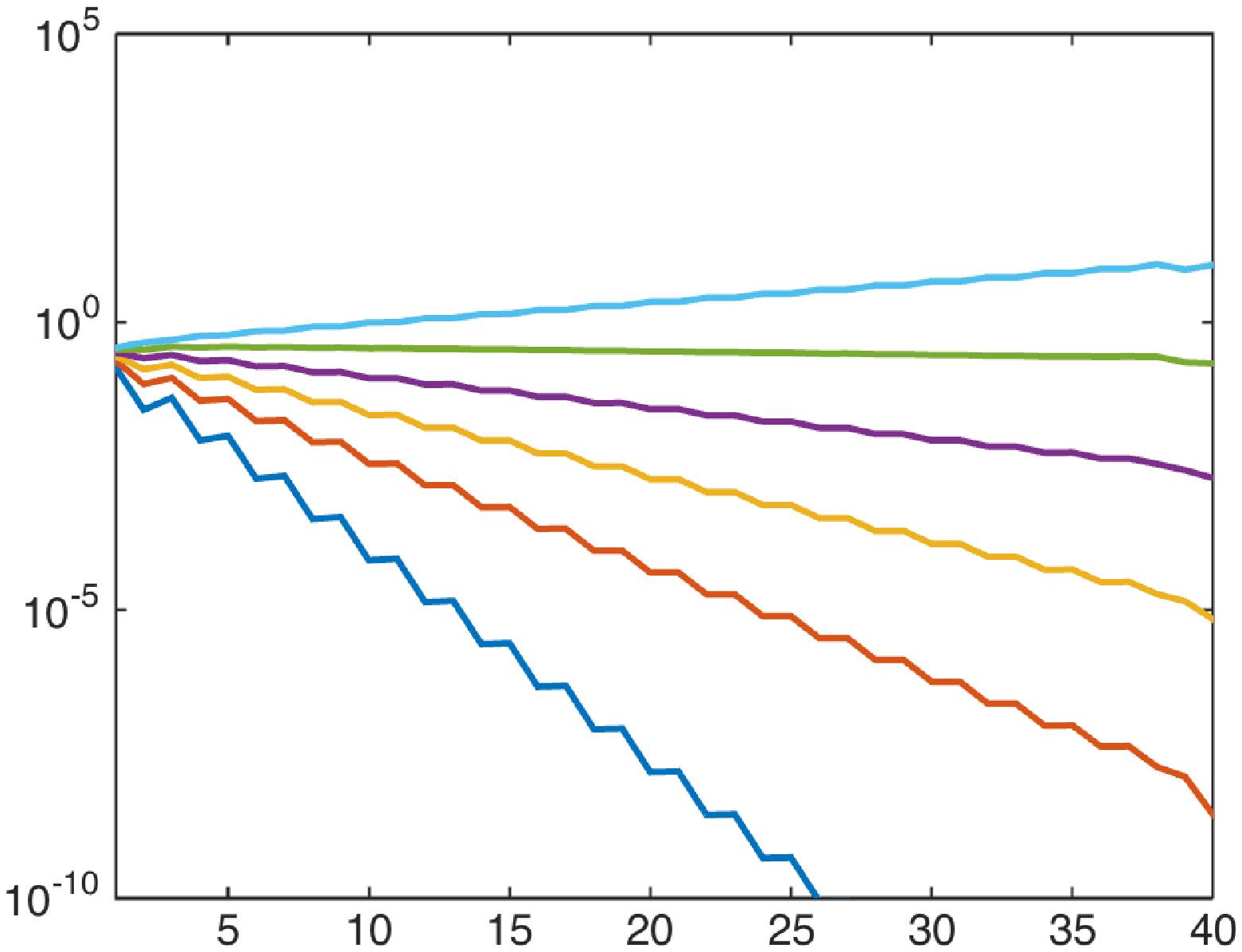}
   \put(50,0) {$M$}
   \put(0,15) {\rotatebox{90}{Extrapolation error}}
   \put(13,64) {$f(x) =1/(1+x^2)$}
  \put(53, 18) {\rotatebox{-37}{$x=1$}}
  \put(70, 25) {\rotatebox{-20}{$x=1.1$}}
  \put(70, 35) {\rotatebox{-14}{$x=1.2$}}
  \put(70, 41) {\rotatebox{-5}{$x=1.3$}}
  \put(70, 47) {\rotatebox{0}{$x=1.4$}}
  \put(70, 53) {\rotatebox{5}{$x=1.5$}}
  \put(14,17) {Decay for}
  \put(14,10) {$1\leq x<1.42$}
\end{overpic} 
\end{minipage}
\begin{minipage}{.49\textwidth} 
 \begin{overpic}[width=\textwidth]{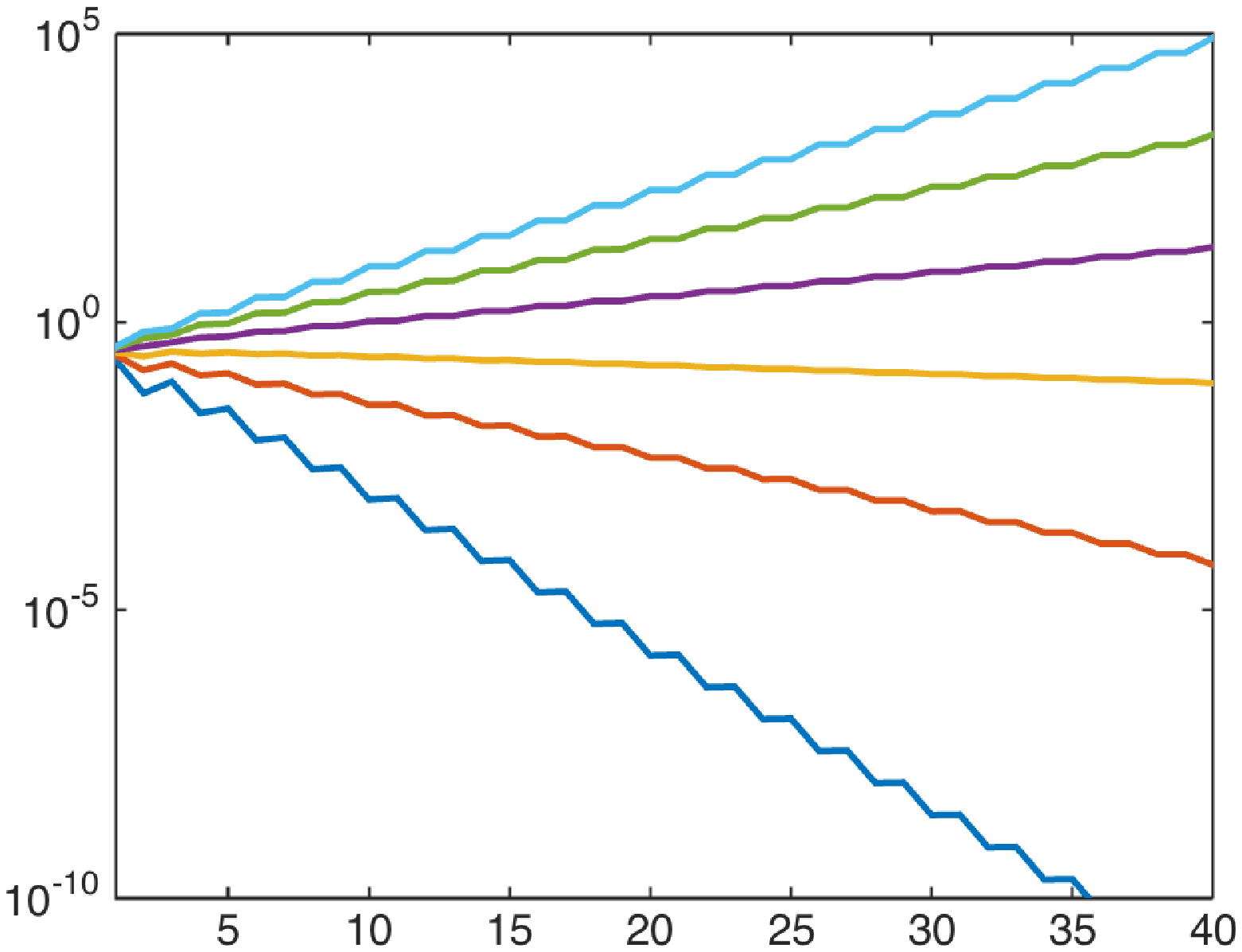}
  \put(50,0) {$M$}
  \put(0,15) {\rotatebox{90}{Extrapolation error}}
  \put(13,64) {$g(x) =1/(1+2x^2)$}
  \put(60, 22) {\rotatebox{-27}{$x=1$}}
  \put(65, 38) {\rotatebox{-11}{$x=1.1$}}
  \put(70, 46) {\rotatebox{-2}{$x=1.2$}}
  \put(70, 53) {\rotatebox{5}{$x=1.3$}}
  \put(70, 59) {\rotatebox{12}{$x=1.4$}}
  \put(60, 61) {\rotatebox{17}{$x=1.5$}}
  \put(14,17) {Decay for}
  \put(14,10) {$1\leq x<1.23$}
 \end{overpic}
\end{minipage}
\caption{Least squares extrapolation error. Left: The error $|f(x) - p_M(x)|$ at $x = 1,1.1,1.2,1.3,1.4,1.5$, 
where $f(x) =1/(1+x^2)$, $p_M$ is the least squares polynomial fit of degree $M$, and $1\leq M\leq 40$. For any $1<x<\sqrt{2}$ the extrapolation error $|f(x) - p_M(x)|$ 
converges geometrically with order $r = (x + \sqrt{x^2-1})/\rho$ to zero as $M\rightarrow\infty$ (see Theorem~\ref{thm:extrapolation}).
Right: The same as the left figure for $g(x) =1/(1+2x^2)$. For any $1< x<1.23$ the extrapolation error $|g(x) - p_M(x)|$ 
converges geometrically to zero as $M\rightarrow\infty$, where $p_M$ is the least squares polynomial fit to $g$ of degree $M$.}
\label{fig:StableExtrapolation} 
\end{figure}

\subsection{Extrapolation with Gaussian noise}\label{sec:ExtrapolationWithNoise}
In the presence of noise in the function samples one must be a little more careful. 
Suppose that the functions samples, $f(\underline{x}^{equi})$, are 
perturbed by noise, $f(\underline{x}^{equi})+\underline{\eps}$, so that each entry of 
$\underline{\eps}$ is modeled by a Gaussian random variable with mean $0$ and variable $s^2$. 
Then, the expected extrapolation error can be bounded as follows: 
\begin{corollary} 
Suppose that the assumptions in Corollary~\ref{cor:ApproximationPowerWithNoise} hold. 
Then, for any  $1 \leq x < (\rho+\rho^{-1})/2$ we have
\[
\begin{aligned}
 \mathbb{E}\left[|f(x) - p_M(x)|\right] &\leq  2Q\left[\frac{(N+1)^{1/2}(M+1)}{\sigma_{M+1}(\mathbf{T}_M(\underline{x}^{equi}))(\rho-1)} + \frac{r}{1-r}\right]r^M\\
 &\qquad\qquad\qquad\qquad\qquad\qquad+\frac{(M+1)^{3/2}s}{\sigma_{M+1}(\mathbf{T}_M(\underline{x}^{equi}))}(\rho r)^M,
\end{aligned}
\]
where $r = (x+\sqrt{x^2-1})/\rho$.
\label{cor:NoisyErrorBound}
\end{corollary}
\begin{proof} 
Essentially the same proof as Theorem~\ref{thm:extrapolation} with an additional term that is bounded 
using Lemma~\ref{lem:varianceNoise}.
\end{proof}

This shows that extrapolation with noise is unstable since $\rho r>1$ and hence, $(\rho r)^M$ grows 
exponentially with $M$. A closer look reveals a more interesting phenomenon though. When $M\leq\constant \sqrt{N}$, using 
Corollary~\ref{cor:NoisyErrorBound}, Theorem~\ref{thm:ChebyshevSingularValues}, and forgetting quantities that grow like a 
polynomial in $M$, we have
\[
 \mathbb{E}\left[|f(x) - p_M(x)|\right] \lesssim \frac{2Qr^{M+1}}{1-r}, \qquad 1 \leq x < (\rho+\rho^{-1})/2 + \frac{s(\rho r)^M}{\sqrt{N}},
\]
where $r=(x+\sqrt{x^2-1})/\rho<1$. Therefore, if the function is exponentially oversampled, i.e., $N \geq e^{aM}$ for some constant $a$, 
then 
\[
 \mathbb{E}\left[|f(x) - p_M(x)|\right] \lesssim \frac{2Qr^{M+1}}{1-r} + s N^{\frac{1}{a}\log(\rho r)-1/2},
\]
and provided that $a > 2\log(\rho r)$ the expected extrapolation error decays to $0$ as $M\rightarrow\infty$. 
This regime may not be as practical as one might hope because exponential oversampling is quite prohibitive; however, 
it reveals that polynomial extrapolation can not only be stable, but also arbitrarily accurate, with function samples perturbed by Gaussian noise.

\subsection{Extrapolation with deterministic perturbations}\label{sec:ExtrapolationWithDeterministicNoise}
A quite different situation occurs when the function samples are perturbed deterministically. 
That is, we obtain function samples of the form $f(\underline{x}^{equi})+\underline{\eps}$ with
$\|\underline{\eps}\|_\infty<\infty = \eps <\infty$. Here, is the bound that one obtains on the extrapolation error.
\begin{corollary} 
Suppose that the assumptions in Corollary~\ref{cor:deterministicError} hold. 
Then, for any fixed $1 \leq x < (\rho+\rho^{-1})/2$ we have
\begin{equation}
\begin{aligned}
 |f(x) - p_M(x)| &\leq 2Q\left[\frac{(N+1)^{1/2}(M+1)}{\sigma_{M+1}(\mathbf{T}_M(\underline{x}^{equi}))(\rho-1)} + \frac{r}{1-r}\right]r^M\\
 &\qquad\qquad\qquad\qquad\qquad\qquad + \frac{(M+1)(N+1)^{1/2}\eps}{\sigma_{M+1}(\mathbf{T}_M(\underline{x}^{equi}))}(\rho r)^M,
\end{aligned}
\label{eq:DeterministicErrorBound}
\end{equation}
where $r = (x+\sqrt{x^2-1})/\rho$.
\end{corollary}
\begin{proof} 
 Essentially the same proof as Theorem~\ref{thm:extrapolation} with an 
 additional term depending on $\eps$ that is relatively simple to bound.
\end{proof}

Since $\rho r>1$, the upper bound in~\eqref{eq:DeterministicErrorBound} does not decay to 
zero as $M\rightarrow\infty$. However, there is again a more interesting phenomenon here to investigate. 
Given an $1 \leq x < (\rho+\rho^{-1})/2$ and a perturbation level $\eps$, we 
can select an integer $M$ that (nearly) minimizes the bound in~\eqref{eq:DeterministicErrorBound}. 
Under the assumption that $M\leq\constant\sqrt{N}$, using Theorem~\ref{thm:ChebyshevSingularValues},
and by ignoring quantities that grow like a polynomial in $M$ (and otherwise depend on $\rho$), we have
\begin{equation}
  |f(x) - p_M(x)| \lesssim  \frac{Qr}{1-r} r^M + (\rho r)^M\eps. 
\label{eq:LastBound}
\end{equation} 

We now turn to the proof of Theorem~\ref{teo:main}. We wish to find an integer $\tilde{M}$ that 
approximately balances the orders of magnitude of the two terms in~\eqref{eq:LastBound}. A simple choice is
\begin{equation}
\tilde{M} = \lfloor \log(Q/\eps)/\log\rho \rfloor.
\label{eq:OptimalM}
\end{equation}
In this case, we get
\begin{equation}
|f(x) - p_{\tilde{M}}(x)| \lesssim \frac{Q}{1-r} \left( \frac{\|\underline{\eps}\|_\infty}{Q} \right)^{-\log r/\log\rho}.
\label{eq:Lastlastbound}
\end{equation} 
Notice that the integer rounding that occurs in~\eqref{eq:OptimalM} only contributes at most a factor $\rho$ to the bound in~\eqref{eq:Lastlastbound} and this 
is absorbed in the constant. We are now ready to prove Theorem~\ref{teo:main}. 

In the oversampled case, i.e., $\tilde{M} < \frac{1}{2} \sqrt{N}$, we can let $M^* = \tilde{M}$, and the bound in~\eqref{eq:Lastlastbound} is the desired result
from Theorem~\ref{teo:main}. 

In the undersampled case, i.e., $\tilde{M} \geq \frac{1}{2} \sqrt{N}$, the value of $\tilde{M}$ is too large to be admissible, so we let $\smash{M^* = \tfrac{1}{2} \sqrt{N}}$ instead. In this case, the term $(Qr/(1-r)) r^M$ dominates in equation~\eqref{eq:LastBound}, and we get
\begin{equation}
|f(x) - p_M(x)| \lesssim \frac{Q}{1-r} r^{\frac{1}{2} \sqrt{N}}.
\label{eq:Lastlastbound-under}
\end{equation}
This concludes the proof of Theorem \ref{teo:main}.  

\subsection{Minimax rate for extrapolation with deterministic perturbations}
One may wonder if it is possible to construct a more accurate extrapolant from perturbed equally spaced samples 
with piecewise polynomials, rational functions, or some other procedure.  Here, we turn our attention to the proof 
of Proposition~\ref{teo:minimax}, which shows that this is not possible.  We achieve this by constructing an analytic function 
$g(x)$ such that $\sup_{x\in[-1,1]} |g(x)|\leq\eps$ and grows as fast as possible for $x>1$.  Any extrapolation procedure cannot distinguish between 
$g(x)$ and the zero function (because function values can be perturbed by $\eps$) and 
therefore, no extrapolation procedure can deliver an accuracy better than $|g(x)|/2$ at $x\in I_\rho$, for both
$g$ and the zero function simultaneously.

Consider the function defined by 
\[
g(x) = \frac{\rho - 1}{\rho} \sum_{n \geq K} \rho^{-n} T_n(x),\qquad K = \lfloor \log(1/\eps)/\log\rho \rfloor.
\]
For $x \in [-1,1]$, it is simple to bound $g(x)$ as follows:
\[
|g(x)| \leq \frac{\rho - 1}{\rho} \sum_{n \geq K} \rho^{-n} = \rho^{-K-1} \leq \eps. 
\]
To formulate a lower bound on $|g(x)|$ for $x \geq 1$, it is helpful to make use of the ``partial generating function" given by
\[
\sum_{n \geq K} \rho^{-n} T_n(x) = \frac{\rho^{-K-1} T_{K+1}(x) - \rho^{-K-2} T_K(x)}{1 - 2 \rho^{-1} x + \rho^{-2}}, \qquad K\geq 1,
\]
which can easily be proved by induction on $K$. The denominator can also be written as 
$1 - 2 \rho^{-1} x + \rho^{-2} = 2\rho^{-1} ( \tfrac{1}{2}(\rho + \rho^{-1}) - x )$,
which readily shows that $f \in B_{\rho'}(Q')$ for every $\smash{\rho' < \rho}$, and for some $\smash{Q' > 0}$. We can now let $\smash{\rho r = x + \sqrt{x^2 - 1}}$, and use the formula $T_n(x) = ((\rho r)^n + (\rho r)^{-n})/2$ to obtain
\begin{align*}
2 \frac{\rho^{K+2}}{\rho - 1} (1 - 2 \rho^{-1} x + \rho^{-2}) g(x) &= (\rho r)^{K+1} + (\rho r)^{-(K+1)} - \rho^{-1} (\rho r)^K - \rho^{-1} (\rho r)^{-K} \\ 
&= (\rho r - \rho^{-1}) (\rho r)^K + ((\rho r)^{-1} - \rho^{-1}) (\rho r)^{-K} \\
&\geq (1 - \rho^{-1}) (\rho r)^K,
\end{align*}
where in the last inequality we used $1 \leq \rho r \leq \rho$.
Next, it is easy to see that
\[
1 - 2 \rho^{-1} x + \rho^{-2} = (1 - \rho^{-1} \rho_{+}) (1 - \rho^{-1} \rho_{-}),
\]
where $\rho_{\pm} = x \pm \sqrt{x^2 - 1}$. We have $\rho^{-1} \rho_{+} = r$, while $0 \leq \rho^{-1} \rho_{-} \leq \rho^{-1}$, so 
\[
1 - 2 \rho^{-1} x + \rho^{-2} \leq 1-r.
\]
Therefore, we conclude that 
\[
g(x) \geq \frac{\rho^{-2}(1-\rho^{-1}) (\rho - 1)}{2} \frac{r^K}{1-r} \equiv c_\rho \frac{r^K}{1-r},\qquad 1\leq x< \frac{\rho+\rho^{-1}}{2},
\]
where $c_\rho$ is a constant that only depends on $\rho$.  By recalling that the value of $K$ is $\lfloor \log(1/\eps)/\log\rho \rfloor$, we obtain
\[
g(x) \geq c_\rho \frac{1}{1-r} \eps^{- \log r/\log\rho}.
\]
This completes the proof of Proposition~\ref{teo:minimax}.



\section{A faster algorithm for equally spaced least squares polynomial fitting}\label{sec:LeastSquaresAlgorithm}
While conducting numerical experiments for this paper, we derived a faster direct 
algorithm for constructing the normal equations. We describe this algorithm now. 

When $M<N$, the least squares problem in~\eqref{eq:leastSquares} is solved 
by the normal equations in~\eqref{eq:normalEquations}, which is an 
$(M+1)\times (M+1)$ linear system for the 
Chebyshev coefficients of $p_M(x)$. 
Since $\mathbf{T}_M(\underline{x}^{equi})$ is an $(N+1)\times (M+1)$ matrix it naively
costs $\mathcal{O}(M^2N)$ operations to compute the matrix-matrix product 
$\mathbf{T}_M(\underline{x}^{equi})^*\mathbf{T}_M(\underline{x}^{equi})$, 
$\mathcal{O}(MN)$ operations to compute the matrix-vector $\mathbf{T}_M(\underline{x}^{equi})^*\underline{f}$,
and $\mathcal{O}(M^3)$ operations to solve the resulting linear system. 
In this section,  we show how the matrix-matrix 
product $\mathbf{T}_M(\underline{x}^{equi})^*\mathbf{T}_M(\underline{x}^{equi})$ can 
be computed in just $\mathcal{O}(M^3)$ operations. When $M = \lfloor \constant\sqrt{N}\rfloor$
this is a computational saving as it reduces $\mathcal{O}(N^2)$ operations 
to construct and solve the normal equations to $\mathcal{O}(N^{3/2})$ operations.

This is a direct algorithm for constructing and solving the normal equations. Alternatively, 
one may use an iterative method such as the conjugate gradient method on the normal equations, where the nonuniform FFT is employed 
to apply the matrix $\mathbf{T}_M(\underline{x}^{equi})^*\mathbf{T}_M(\underline{x}^{equi})$ to a vector in 
$\mathcal{O}(N\log N)$ operations. While each iteration is fast, the condition number 
of $\mathbf{T}_M(\underline{x}^{equi})^*\mathbf{T}_M(\underline{x}^{equi})$
is $\mathcal{O}(M)$ (see~\ref{eq:ChebyshevConditionNumber}) so one expects the conjugate gradient method 
to require about $\mathcal{O}(M^{1/2})$ iterations. Hence, the 
algorithmic complexity of the iteration approach is $\mathcal{O}(M^{1/2}N\log N)$ operations.  
We prefer the direct approach because the cost of constructing the normal equations is independent 
of $N$.

Our key observation is that the $(m,n)$ entry of $\mathbf{T}_M(\underline{x}^{equi})^*\mathbf{T}_M(\underline{x}^{equi})$
given by $\sum_{k=0}^N T_m(x_k^{equi})T_n(x_k^{equi})$, which can be thought of as 
a trapezium rule approximation of an integral. To see this recall that for a 
continuous function $h(x)$ the trapezium rule approximation to its integral is 
\begin{equation}
 \int_{-1}^1 h(x)dx \approx \frac{1}{N}\left(h(x_0^{equi}) + 2h(x_1^{equi})+\cdots+2h(x_{N-1}^{equi}) + h(x_N^{equi}) \right) 
\label{eq:trapeziumRule}
\end{equation} 
and hence, we have for $0\leq m,n\leq M$
\begin{equation}
\int_{-1}^1 T_m(x)T_n(x)dx =   \frac{2}{N}\sum_{k=0}^N T_m(x_k^{equi})T_n(x_k^{equi}) - \frac{1}{N}(1+(-1)^{m+n}) - \tilde{E}_{mn},
\label{eq:integralTmTn}
\end{equation} 
where $\smash{\tilde{E}_{mn}}$ is the error in the trapezium rule approximation. After rearranging, calculating the integral in~\eqref{eq:integralTmTn} 
analytically, and noting that $\sum_{k=0}^N T_m(x_k^{equi})T_n(x_k^{equi}) = 0$ if $m+n$ is odd, we conclude that 
\begin{equation}
 \sum_{k=0}^N T_m(x_k^{equi})T_n(x_k^{equi}) = \begin{cases} \frac{N}{2(1-(m+n)^2)} +\frac{N}{2(1-(m-n)^2)} + 1 +  \frac{N}{2}\tilde{E}_{mn}, & m+n \text{ is even}, \cr 0, & \text{otherwise}.  \end{cases}
\label{eq:analyticTmTn}
\end{equation}

The sum on the lefthand side of~\eqref{eq:analyticTmTn}, which is used when 
naively computing the $(m,n)$ entry of $\mathbf{T}_M(\underline{x}^{equi})^*\mathbf{T}_M(\underline{x}^{equi})$,
costs $\mathcal{O}(N)$ operations to evaluate. While 
the righthand side requires $\mathcal{O}(M)$ operations because of the 
Euler--Maclaurin error formula for $\tilde{E}_{mn}$~\cite[Cor.~3.3]{Javed_14_01}. 
That is, 
\[
\tilde{E}_{mn} = 2\sum_{s=1,s \text{ odd}}^{m+n-1} \frac{((T_m(1)T_n(1))^{(s)} - (T_m(-1)T_n(-1))^{(s)})2^sB_{s+1}}{N^{s+1}(s+1)!},
\]
where $B_{s}$ is the $s$th Bernoulli number and $(T_m(1)T_n(1))^{(s)}$ is the $s$th 
derivative of $T_m(x)T_n(x)$ evaluated at $1$. By calculating $(T_m(\pm1)T_n(\pm1))^{(s)}$ analytically and rearranging 
we have
\begin{equation}
 \tilde{E}_{mn} = \frac{2}{N} \sum_{s=1,s \text{ odd}}^{m+n-1} \left[\prod_{j=0}^{s-1} \frac{(m-n)^2-j^2}{N(j+1/2)} + \prod_{j=0}^{s-1} \frac{(m+n)^2-j^2}{N(j+1/2)}\right]\frac{B_{s+1}}{(s+1)!}. 
\label{eq:bernoulli}
\end{equation} 
Here, the summand contains two products. The first product depends on $m-n$ and 
the other on $m+n$, in a Toeplitz-plus-Hankel structure. For $0\leq m,n\leq M$ 
there are only $M+1$ possible values for $(m-n)^2$, $2M+1$ possible 
values of $(m+n)^2$, and at most $2M-1$ values of $1\leq s\leq m+n-1$.  This 
means that there are $\mathcal{O}(M^2)$ different products that appear in the set
of formulas for $\tilde{E}_{mn}$, $0\leq m,n\leq M$, and these can be computed in a 
total of $\mathcal{O}(M^2)$ operations.  In principle each $\tilde{E}_{mn}$ for $0\leq m,n\leq M$ 
sums up $\mathcal{O}(M)$ of these products weighted by Bernoulli numbers, requiring
a total of $\mathcal{O}(M^3)$ operations to compute $\tilde{E}$. However, the weighted Bernoulli 
numbers decay so rapidly to zero that we truncate the sums in~\eqref{eq:bernoulli} if $s>10$. Therefore, 
we can compute $\tilde{E}$ in $\mathcal{O}(M^2)$ operations.
Once the matrix $\tilde{E}$ is calculated the matrix $\mathbf{T}_M(\underline{x}^{equi})^*\mathbf{T}_M(\underline{x}^{equi})$ 
can immediately computed from~\eqref{eq:analyticTmTn}. 

For the Bernoulli numbers in~\eqref{eq:bernoulli} we tabulate
$B_{s+1}/(s+1)!$ for $2\leq s\leq 10$ and then use the first six terms in an
asymptotic expansion for $s>10$, i.e.,
\[
 \frac{B_{s+1}}{(s+1)!} \approx (-1)^{(s+3)/2}\frac{2}{(2\pi)^{s+1}}\left(1 + \frac{1}{2^{s+1}} +\frac{1}{3^{s+1}} + \frac{1}{4^{s+1}} + \frac{1}{5^{s+1}}+ \frac{1}{6^{s+1}}\right),\quad \text{$s$ odd}.
\]
This alleviates overflow issues with computing $B_{s+1}$ and $(s+1)!$ separately when $s$ is large.  

Figure~\ref{fig:ComputationalTimings} shows the computational timings for constructing 
the normal equations, 
$\mathbf{T}_M(\underline{x}^{equi})^*\mathbf{T}_M(\underline{x}^{equi})\underline{c} = \mathbf{T}_M(\underline{x}^{equi})^*\underline{f}$, 
using the naive approach and the algorithm described in this section. When $M=\lfloor\constant\sqrt{N}\rfloor$
and $M>50$, it is computationally more efficient to construct the normal equations using this new direct algorithm. 
\begin{figure} 
\centering
 \begin{overpic}[width=.6\textwidth]{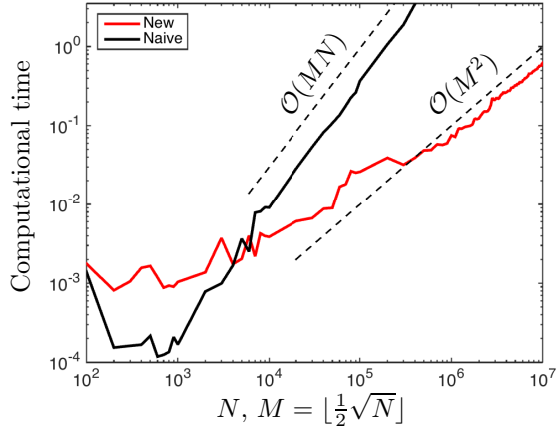}
  \put(0,20) {\rotatebox{90}{Computational time}}
  \put(35,-1) {$N$, $M = \lfloor \constant\sqrt{N}\rfloor$}
  \put(45,50) {\rotatebox{55}{$\mathcal{O}(MN)$}}
  \put(70,50) {\rotatebox{40}{$\mathcal{O}(M^2)$}}
 \end{overpic}
 \caption{Computational times for constructing the normal equations in~\eqref{eq:normalEquations} 
 when $M = \lfloor \constant\sqrt{N}\rfloor$. We compare the naive approach (black) 
 that directly computes the matrix-matrix product $\mathbf{T}_M(\underline{x}^{equi})^*\mathbf{T}_M(\underline{x}^{equi})$ and 
 the approach (red) described in Section~\ref{sec:LeastSquaresAlgorithm}. While the naive approach 
 has a theoretical complexity of $\mathcal{O}(M^2N)$ when $M = \lfloor \constant\sqrt{N}\rfloor$, the 
 dominating computational cost for $N<10^6$ is the $\mathcal{O}(MN)$ cost of 
 evaluating $\mathbf{T}_M(\underline{x}^{equi})$ using a three-term recurrence.} 
 \label{fig:ComputationalTimings}
\end{figure}


\section*{Acknowledgments}
We wish to thank Mohsin Javed for his correspondence regarding the Euler--Maclaurin error formula 
in~\cite{Javed_14_01}. We also thank Ben Adcock for directing us to the literature
on stable reconstruction and telling us about~\cite{Adcock_15_02}.  LD is grateful 
to AFOSR, ONR, NSF, and Total SA for funding.

\appendix 

\section{Three applications of Gerschgorin's circle Theorem}\label{sec:Appendix}
Gerschgorin's circle Theorem can be used to bound the spectrum of a square
matrix as it restricts the eigenvalues of a matrix $A$ to the union 
of disks centered at the diagonal entries of $A$~\cite[p.~320]{Golub_96_01}. 
\begin{theorem}
 Let $A\in\mathbb{C}^{n\times n}$ with entries $a_{ij}$. Then, 
 the eigenvalues of $A$ lie within at least one of the {\em Gerschgorin 
 disks},
 \[
  \left|z-a_{ii}\right|\leq \sum_{j=1,j\neq i}^n |a_{ij}|,\qquad 1\leq i\leq n. 
 \] 
 \label{thm:Gerschgorin}
\end{theorem}

For a given square matrix $A$, the eigenvalue bounds given in 
Theorem~\ref{thm:Gerschgorin} can be quite weak. A standard trick 
is to sharpen the bounds by using a carefully selected similarity 
transform. For any invertible matrix $P$ the matrix $PAP^{-1}$ has the 
same spectrum as $A$, but may have Gerschgorin disks with smaller radii and 
this can sharpen a bound on an eigenvalue of interest. Here, we apply Gerschgorin's circle Theorem to three matrices and 
select diagonal similarity transforms to improve the bounds. 

First, we use Gerschogrin's circle Theorem to bound the spectrum of the matrix $D+C$
from Theorem~\ref{thm:LegendreVandermondeSingularValues}. This result is used 
to then derive a bound on the singular values of $\mathbf{P}_M(\underline{x}^{equi})$. 
In Figure~\ref{fig:Gcircles} we draw the Gerschgorin circles for $D+C$ and 
$P (D+C) P^{-1}$, where $P = {\rm diag}(D_{00},\ldots,D_{MM})$. It is this diagram 
that motivates the proof of the lemma below. 
\begin{figure} 
\begin{minipage}{.49\textwidth} 
\begin{overpic}[width=\textwidth]{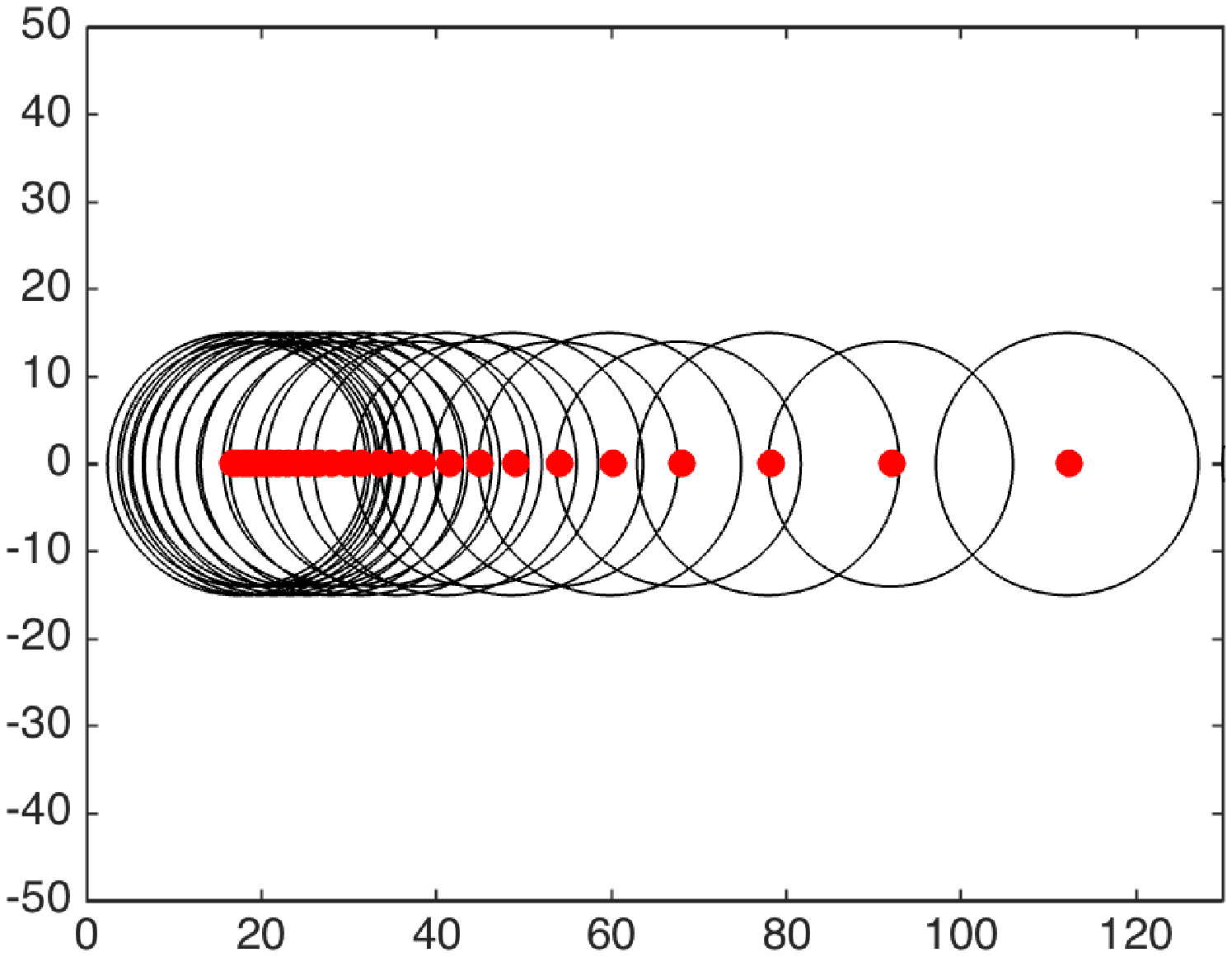}
 \put(47,0) {${\rm Re}(z)$}
 \put(0,30) {\rotatebox{90}{${\rm Im}(z)$}}
\end{overpic}
\end{minipage}
\begin{minipage}{.49\textwidth} 
\begin{overpic}[width=\textwidth]{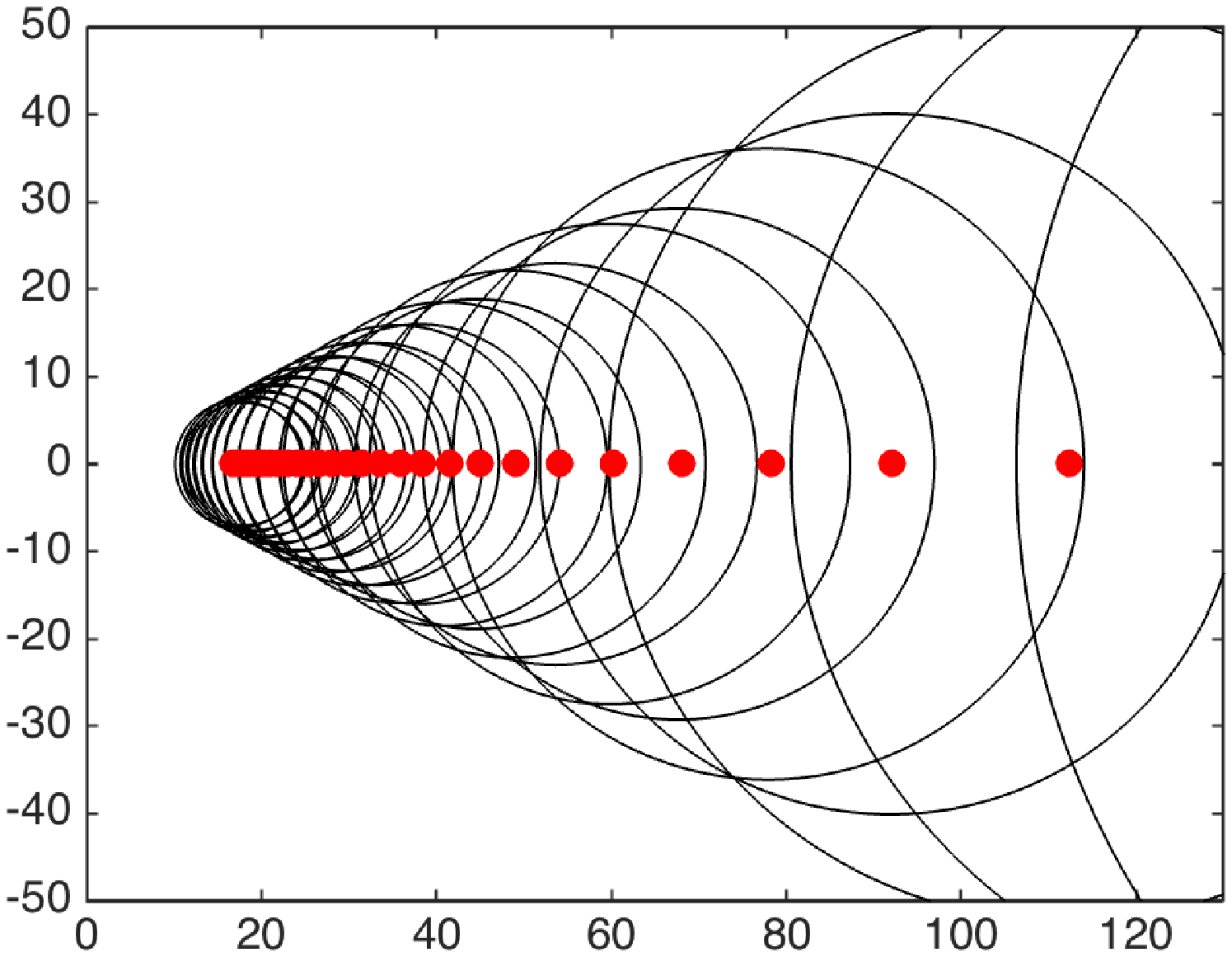}
  \put(47,0) {${\rm Re}(z)$}
 \put(0,30) {\rotatebox{90}{${\rm Im}(z)$}}
\end{overpic}
\end{minipage}
\caption{ Left:   The Gerschgorin disks for the matrix $D + C$ near ${\rm Re}(z)=0$ used in 
Lemma~\ref{lem:GcirclesDplusC} when $N = 1,\!000$ and $M = 30$.
Without a similarity transform the Gerschgorin circles give a poor lower bound on 
$\lambda_{M+1}(D+C)$. Right: The Gerschgorin disks for $P(D+C)P^{-1}$ near ${\rm Re}(z)=0$, where 
$P = {\rm diag}(D_{00},\ldots,D_{MM})$. The Gerschgorin disks now give a 
better lower bound on $\lambda_{M+1}(D+C)$. Another diagram shows that 
a similarity transform is not needed for bounding $\lambda_1(D+C)$.}
 \label{fig:Gcircles}
\end{figure}

\begin{lemma} 
 For integers $M$ and $N$ satisfying $M\leq N$, let $D$ be the diagonal matrix with entries $D_{mm} = N/(2m+1)$ for $0\leq m\leq M$ 
 and $C$ be the $(M+1)\times (M+1)$ matrix given in~\eqref{eq:Cmatrix}.  
 The following bounds on the maximum and minimum eigenvalues hold:
\[
 \lambda_1(D+C)\leq \frac{2N+M+3}{2}, \qquad \lambda_{M+1}(D+C)\geq \frac{N-\tfrac{1}{2}M^2}{2M+1}.
\]
\label{lem:GcirclesDplusC}
\end{lemma}
\begin{proof} 
The matrix $D+C$ is symmetric so all the eigenvalues are real. 
By Theorem~\ref{thm:Gerschgorin} applied to $D+C$ (without a similarity transform) 
we find that 
\[
 \lambda_1(D+C) \leq \max_{0\leq j\leq M} \left\{(D+C)_{jj} + \sum_{k=0,k\neq j}^M |C_{jk}| \right\}\leq N + 1 + \tfrac{M+1}{2},
\]
as required. For $\lambda_{M+1}$ we consider the matrix $P (D+C) P^{-1}$, 
where $P$ is the diagonal matrix ${\rm diag}(D_{00},\ldots,D_{MM})$. By Theorem~\ref{thm:Gerschgorin} we have
\[
\begin{aligned}
 \lambda_{M+1}(D+C) &\geq \min_{0\leq j\leq M} \left\{(P(D+C)P^{-1})_{jj} - \sum_{k=0,k\neq j}^M |(PCP^{-1})_{jk}| \right\}\\
 &\geq \frac{N}{2M+1} + 1 - \frac{N}{2M+1}\frac{(M+1)(M+2)}{2N}\\
 &\geq \frac{N-\tfrac{1}{2}M^2}{2M+1} + \frac{M}{2(2M+1)} \geq \frac{N-\tfrac{1}{2}M^2}{2M+1},
\end{aligned}
\]
as required.
\end{proof}

The second application of Gerschgorin's circle Theorem is on the matrix 
$F + C$ appearing in Theorem~\ref{thm:ChebyshevSingularValues}, where an upper 
bound on the maximum eigenvalue of $F+C$ is required. A similarity 
transform is not needed here.
\begin{lemma} 
 Let $M$ and $N$ be integers satisfying $M\leq N$. Let $F$ be the matrix given in~\eqref{eq:F} 
 and $C$ be the $(M+1)\times (M+1)$ matrix given in~\eqref{eq:Cmatrix}.  
 The following bound on the maximum eigenvalue holds:
\[
 \lambda_1(F+C)\leq \frac{4N+M+1}{2}.
\]
\label{lem:GcirclesDplusC2}
\end{lemma}
\begin{proof} 
The matrix $F+C$ is symmetric so all the eigenvalues are real. By Theorem~\ref{thm:Gerschgorin} applied to $D+C$ we 
have 
\[
 \lambda_1(F+C) \leq \max_{0\leq j\leq M} \left\{(F+C)_{jj} + \sum_{k=0,k\neq j}^M |F_{jk} + C_{jk}| \right\}\leq 2N + \frac{M+1}{2},
\]
as required.
\end{proof}

Lemma~\ref{lem:GcirclesDplusC} and Lemma~\ref{lem:GcirclesDplusC2} are 
easy applications of Gerschgorin's circle Theorem; however, the
next application is more technical. For Theorem~\ref{thm:ChebyshevSingularValues}, we want to bound $\|S\|_2$, where $S$ is the change of basis
matrix given in~\eqref{eq:S}. It is also not clear if the Gerschgorin's circle Theorem is 
applicable here. Fortunately, $S$ is a matrix with nonnegative entries so that it is possible to 
bound $\|S\|_2$ by the spectrum of its symmetric part~\cite{Goldberg_75_01}. 

Let $r(S)=\sup\left\{|v^*Sv| : v\in\mathbb{C}^{M\times 1}, v^*v = 1\right\}$ be 
the numerical range of $S$. Then, $\|S\|_2\leq 2r(S)$. Since $S$ has 
nonnegative entries we have~\cite[Thm.~1]{Goldberg_75_01}
\[
r(S) \leq \max\left\{|\lambda| : S^{+}v=\lambda v, v\neq 0\right\},
\]
where $S^{+} = (S+S^*)/2$ is the symmetric part of $S$. Therefore, we can use 
Gerschgorin's circle Theorem to bound $\max_{1\leq i\leq M+1}|\lambda_{i}(S^{+})|$ and then use
\begin{equation}
 \|S\|_2 \leq 2\max_{1\leq i\leq M+1}|\lambda_{i}(S^{+})|.
\label{eq:StrangeBound}
\end{equation} 
It is technical to bound $\max_{1\leq i\leq M+1}|\lambda_{i}(S^{+})|$ using Gerschgorin's circle 
Theorem. In Figure~\ref{fig:GcirclesSplus} we show the Gerschgorin's disk for 
$S^{+}$ and $PS^{+}P^{-1}$, where $P_{00} = 1$ and $P_{ii} = \sqrt{i}$ for $i\geq 1$. 
The circles are tight if we work with $PS^{+}P^{-1}$.

\begin{figure} 
\begin{minipage}{.49\textwidth} 
\begin{overpic}[width=\textwidth]{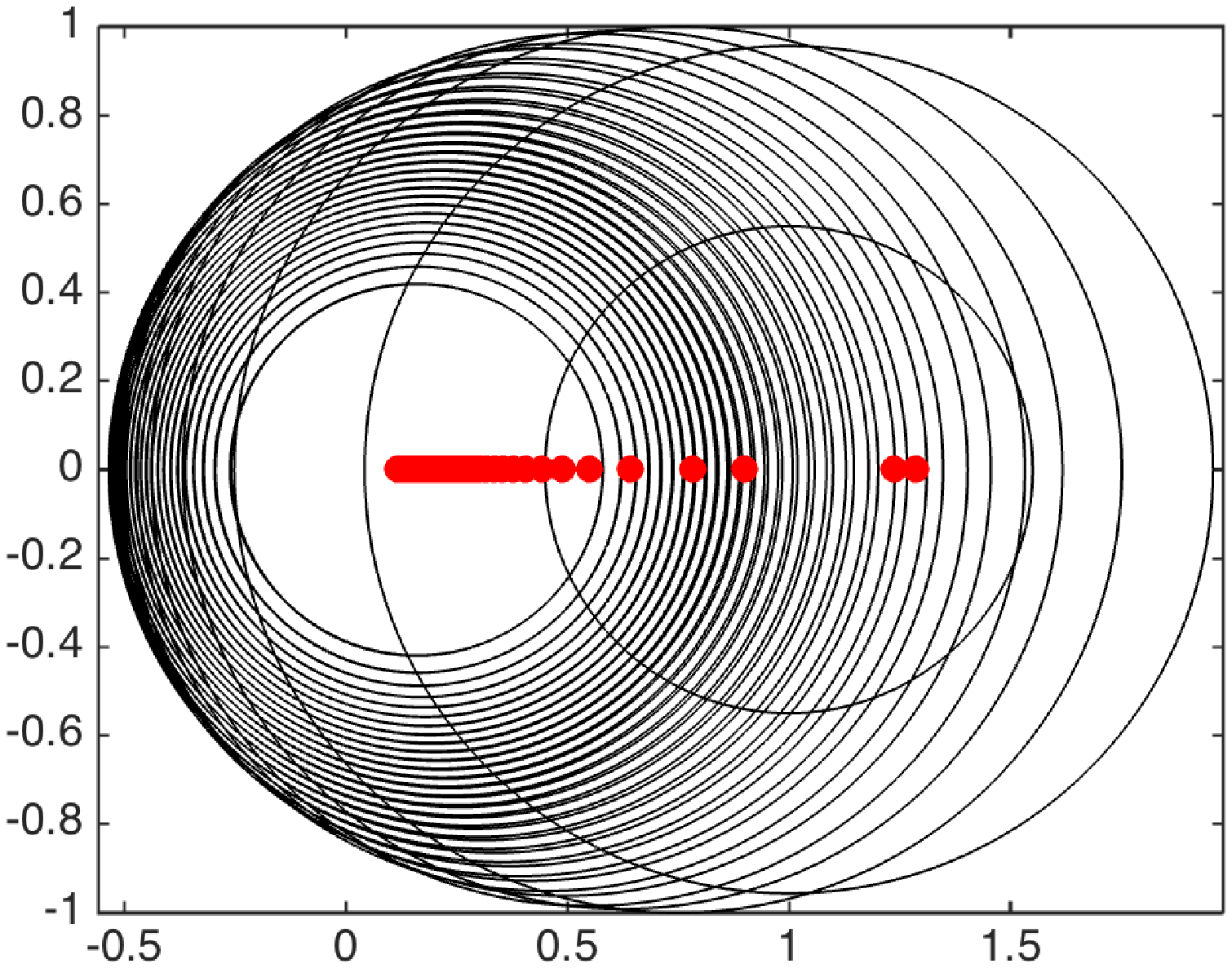}
 \put(47,0) {${\rm Re}(z)$}
 \put(0,30) {\rotatebox{90}{${\rm Im}(z)$}}
\end{overpic}
\end{minipage}
\begin{minipage}{.49\textwidth} 
\begin{overpic}[width=\textwidth]{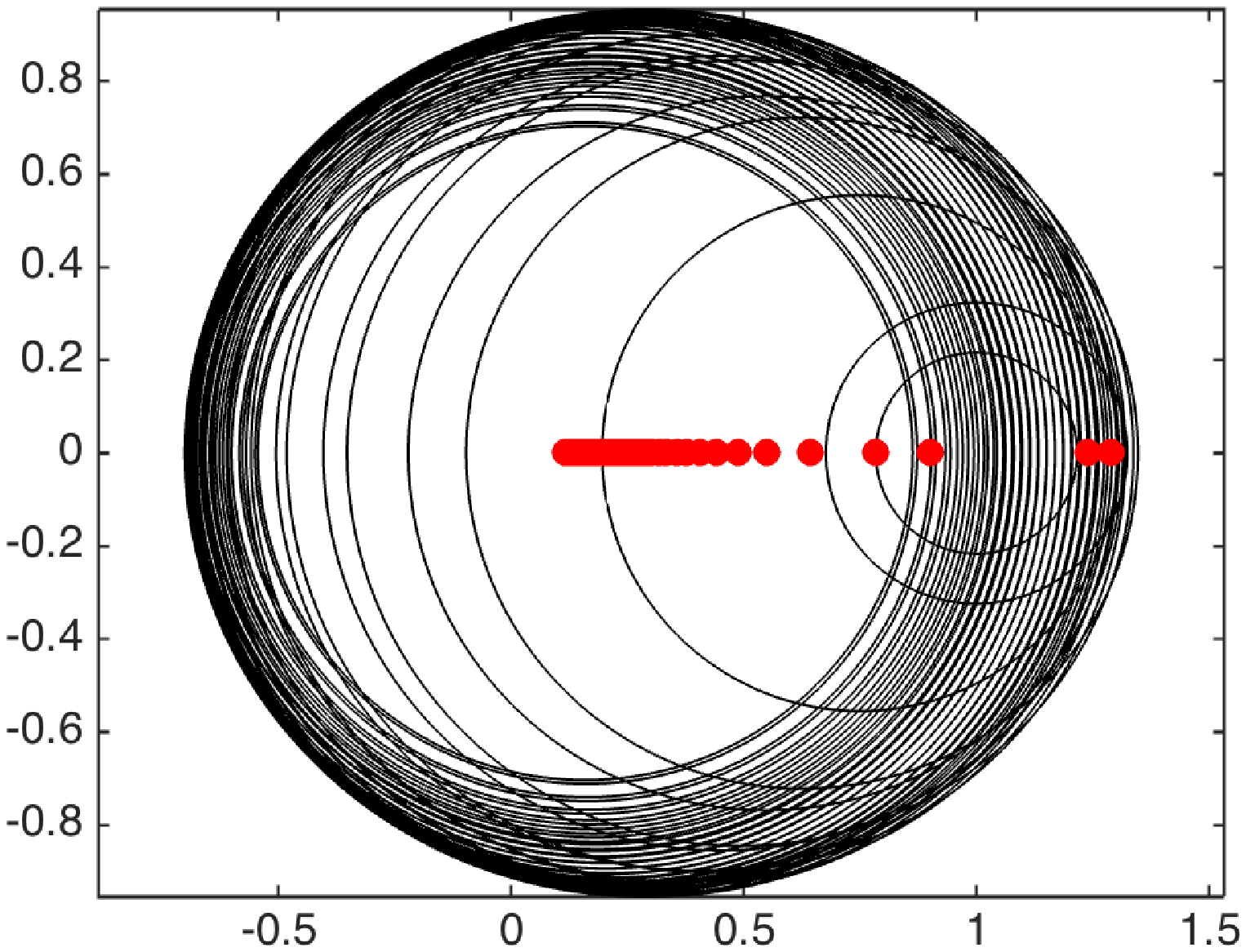}
  \put(47,0) {${\rm Re}(z)$}
 \put(0,30) {\rotatebox{90}{${\rm Im}(z)$}}
\end{overpic}
\end{minipage}
\caption{Left: The Gerschgorin disks for $S^{+}$ in 
Lemma~\ref{lem:BS} when $M = 50$. 
Without a similarity transform the Gerschgorin circles give a poor upper bound on 
$\lambda_{1}(S^{+})$. Right: The Gerschgorin disks for $P(D+C)P^{-1}$, where 
$P = {\rm diag}(1,\sqrt{1},\ldots,\sqrt{M})$. The Gerschgorin disks now provide a
tight upper bound on $\lambda_{1}(S^{+})$ as $M\rightarrow\infty$.}
 \label{fig:GcirclesSplus}
\end{figure}
\begin{lemma} 
 Let $M$ be an integer and $S^{+}$ be the symmetric part of the $(M+1)\times (M+1)$ matrix $S$ in~\eqref{eq:S}. Then, 
 \[
  \max_{1\leq i\leq M+1}|\lambda_{i}(S^{+})|\leq \frac{5}{2}.
 \]
 From~\eqref{eq:StrangeBound} we conclude that $\|S\|_2\leq 5$.
 \label{lem:BS}
\end{lemma}
\begin{proof} 
We apply Theorem~\ref{thm:Gerschgorin} to $A = PS^{+}P^{-1}$, where $P_{00} = 1$ and 
$P_{ii} = \sqrt{i}$ for $i\geq 1$. The entries of $A$ are given explicitly by 
\[
 A_{ij} = 
 \begin{cases}
 1, & i=j=0,\cr
 \tfrac{1}{2\pi\sqrt{j}}\Psi(\tfrac{j}{2})^2,&i=0,j>0, j\text{ even},\cr
 \tfrac{\sqrt{i}}{2\pi}\Psi(\tfrac{i}{2})^2,&j=0,i>0, i\text{ even},\cr
 \tfrac{\sqrt{i}}{\pi\sqrt{j}}\Psi(\tfrac{j-i}{2})\Psi(\tfrac{j+i}{2}),&i,j>0, i+j\text{ even},\cr
 \end{cases} \qquad 0\leq i,j\leq M,
\] 
where $\Psi(j) = \Gamma(j + 1/2)/\Gamma(j + 1)$ and $\Gamma(x)$ is the Gamma function.
We consider the Gerschgorin's disk in four cases: (1) the disk centered at $A_{00}$, (2) the disk
centered at $A_{11}$, (3) the disks centered at $A_{ii}$ with $i=2k>0$; and, (4) 
the disks centered at $A_{ii}$ with $i=2k+1>1$.

\subsubsection*{Case 1: The Gerschogrin disk centered at $\mathbf{A_{00}}$}
First note that by Wendel's lower bound on the ratio of Gamma functions~\cite{Wendel_48_01} 
we have
\begin{equation}
 \Psi(j)^2 \leq \frac{j+1}{(j+1/2)^2} \leq \frac{1}{j}, \qquad j\geq 1.
\label{eq:WendelBound}
\end{equation}
Using~\eqref{eq:WendelBound} we can bound the radius of the Gerschogrin disk centered at $A_{00}$
as follows:
\[
\begin{aligned}
 \sum_{j=1}^{\lfloor M/2 \rfloor} A_{0,2j} \leq \frac{1}{2\pi} \sum_{j=2}^{\infty} \frac{\Psi(j)^2}{\sqrt{2j}} \leq \frac{1}{2\sqrt{2}\pi}\sum_{j=2}^\infty \frac{1}{j^{3/2}} = \frac{1}{2\sqrt{2}\pi}\left(\zeta(3/2)-1\right)\leq 0.19.
\end{aligned}
\]
Since $A_{00} = 1$ the Gerschgorin disk is contained in $\{z\in\mathbb{C} : |z| \leq 1.19\}$.

\subsubsection*{Case 2: The Gerschogrin disk centered at $\mathbf{A_{11}}$}
Since $\Gamma(z+1) = z\Gamma(z)$ we have 
\[
 \Psi( j + 1 ) = \frac{j+1/2}{j+1}\Psi(j) \leq \Psi( j ), \qquad j\geq 0,
\]
and hence, $\Psi(0),\Psi(1),\Psi(2),\ldots$ is a monotonically decreasing sequence.
Using this we can bound the radius of the Gerschogrin disk centered at 
$A_{11}$ as follows:
\[
\begin{aligned}
 \sum_{j=1}^{\lfloor M/2 \rfloor} A_{1,2j-1} &\leq \frac{1}{\pi\sqrt{3}}\Psi(1)\Psi(2) + \frac{1}{\pi} \sum_{j=3}^{\infty} \frac{\Psi(j-1)\Psi(j)}{\sqrt{2j-1}}\\
 &\leq \frac{1}{\pi\sqrt{3}}\Psi(1)\Psi(2) + \frac{\sqrt{2}}{2\pi}(\zeta(3/2)-1)\leq 0.48,
\end{aligned}
\]
where we used $\Psi(j-1)\Psi(j)\leq \Psi(j)^2\leq j^{-1} \leq (j-1)^{-1}$ and $(2j-1)^{-1/2}\leq \sqrt{2} j^{-1/2}$. 
Since $A_{11} = 1$ the Gerschgorin disk is contained in $\{z\in\mathbb{C} :|z| \leq 1.48\}$.

\subsubsection*{Case 3: The Gerschogrin disks centered at $\mathbf{A_{ii}}$ with $\mathbf{i=2k>0}$}  
The radii of the Gerschgorin disks centered at $A_{ii}$ with $i = 2k>0$ is bounded by
\[
\begin{aligned}
 \sum_{j=0,j\neq k}^{\lfloor M/2 \rfloor} A_{2k,2j} \leq \underbrace{\frac{1}{2\pi}\Psi(k)^2(2k)^{\tfrac{1}{2}}}_{=(1)} &+ \underbrace{\frac{1}{\pi} \sum_{j=1}^{k-1} \Psi(k-j)\Psi(k+j)\left(\frac{k}{j}\right)^{\tfrac{1}{2}}}_{=(2)} \\
 &\qquad\qquad\qquad+ \underbrace{\frac{1}{\pi} \sum_{j=k+1}^{\infty} \Psi(j-k)\Psi(j+k)\left(\frac{k}{j}\right)^{\tfrac{1}{2}}}_{=(3)}.
\end{aligned}
\]
We bound the three parts in turn. By~\eqref{eq:WendelBound} we have 
\[
 (1) \leq \frac{1}{2\pi}\Psi(k)^2(2k)^{\tfrac{1}{2}} \leq \frac{\sqrt{2}}{2\pi}k^{-1/2} \leq 0.23.
\]

Next, note that $(2)=0$ if $k=1$ so we can assume that $k\geq 2$. For $k\geq 2$, the summands in $(2)$ 
have a single local minimum. For small $j$ the terms in~$(2)$ 
are monotonically decreasing and for larger $j$ are monotonically increasing. This means we can 
apply a double-sided integral test to bound the sum. That is, 
\begin{equation}
 (2) \leq \frac{1}{\pi}\Psi(1)\Psi(2k-1)\frac{k^{1/2}}{(k-1)^{1/2}} + \frac{1}{\pi}\Psi(k-1)\Psi(k+1)k^{1/2} + \frac{1}{\pi}\int_{1}^{k-1}\!\!\!\!\!\!\!\frac{\sqrt{k}}{\sqrt{(k-x)x(x+k)}}dx.
\label{eq:ohBoy}
\end{equation} 
Since $\Psi(1)=\sqrt{\pi}/2$, $\Psi(2k-1)\leq (2k-1)^{-1/2}$, $\Psi(k-1)\Psi(k+1)k^{1/2}\leq k^{-1/2}$, and the fact 
that the continuous integral in~\eqref{eq:ohBoy} can be expressed in terms of a hypergeometric function, we have 
\begin{equation}
\begin{aligned} 
 (2) &\leq \frac{k^{-1/2}}{\pi} + \frac{1}{2\sqrt{\pi}}\left(\frac{k}{(k-1)(2k-1)}\right)^{\tfrac{1}{2}} + \frac{2}{\pi\sqrt{k}}\,{}_2F_1(\tfrac{1}{4},\tfrac{1}{2},\tfrac{5}{4};(1-k^{-1})^2)\\
 &\leq \frac{1}{\sqrt{2}\pi} + \frac{\sqrt{3}}{2\sqrt{2\pi}} +  \frac{1.03}{2\sqrt{2\pi}} \leq 0.78.
\end{aligned}
\label{eq:2}
\end{equation}
Here, in the penultimate inequality we used $k^{-1/2}\leq 1/\sqrt{2}$ for $k\geq 2$, $k/(k-1)(2k-1)\leq 2/3$ for $k\geq 2$, 
and ${}_2F_1(\tfrac{1}{4},\tfrac{1}{2},\tfrac{5}{4};(1-k^{-1})^2)\leq {}_2F_1(\tfrac{1}{4},\tfrac{1}{2},\tfrac{5}{4};\tfrac{1}{4})\leq 1.03$ for $k\geq 2$. 

Finally, for $(3)$ we note that the summands are monotonically decreasing so that by the integral bound we have
\[ 
 (3) \leq \frac{1}{\pi}\Psi(1)\Psi(2k+1)\left(\frac{k}{k+1}\right)^{\tfrac{1}{2}} + \frac{1}{\pi}\int_{k+1}^\infty \frac{\sqrt{k}}{\sqrt{(x-k)x(x+k)}}dx.
\]
Since $\Psi(1)=\sqrt{\pi}/2$, $\Psi(2k+1)^2\leq 1/(2k+1)$, and the continuous integral can be 
transformed into a elliptic integral (of the first kind), denoted by $F$, we have
\begin{equation} 
\begin{aligned}
 (3) &\leq \frac{1}{2\sqrt{\pi}}\left(\frac{k}{(k+1)(2k+1)}\right)^{1/2} + \frac{2}{\pi}F\left(\sin^{-1}\left(\sqrt{\tfrac{k}{k+1}}\right),-1\right)\\
 &\leq \frac{1}{2\sqrt{6\pi}} + \frac{\Gamma(1/4)^2}{2\sqrt{2}\pi^{3/2}} \leq 0.95.
\end{aligned} 
\label{eq:3}
\end{equation}
Here, in the penultimate inequality we used $k/((k+1)(2k+1))\leq 1/6$ for $k\geq 1$ and 
$F(\sin^{-1}(\sqrt{k/(k+1)}),-1)\leq F(\pi/2,-1)=\Gamma(1/4)^2/(4\sqrt{2\pi})$.

Since $|A_{ii}|\leq A_{22} \leq 3/8$ for $i\geq 2$ and 
$3/8+ 0.23 + 0.78 + 0.95\leq 2.34$, these Gerschgorin disks are 
contained in $\{z\in\mathbb{C} : |z|\leq 2.34\}$. 

\subsubsection*{Case 4: The Gerschogrin disks centered at $\mathbf{A_{ii}}$ with $\mathbf{i=2k+1>1}$}
The radii of a Gerschgorin disk centered at $A_{ii}$ with $i = 2k+1>0$ is bounded by
\[
\begin{aligned}
 \sum_{j=0,j\neq k}^{\lfloor M/2 \rfloor} A_{2k+1,2j+1} &\leq \underbrace{\frac{1}{\pi} \sum_{j=1}^{k-1} \Psi(k-j)\Psi(k+j+1)\left(\frac{2k+1}{2j+1}\right)^{\tfrac{1}{2}}}_{=(i)} \\
 &\qquad\qquad\qquad\qquad+ \underbrace{\frac{1}{\pi} \sum_{j=k+1}^{\infty} \Psi(j-k)\Psi(j+k+1)\left(\frac{2k+1}{2j+1}\right)^{\tfrac{1}{2}}}_{=(ii)}.
\end{aligned}
\]
Since $\Psi(j+k+1)\leq \Psi(j+k)$ and $(2k+1)/(2j+1)\leq k/j$ for $1\leq j\leq k-1$, we have
from~\eqref{eq:2}
\[
 (i)\leq \frac{1}{\pi} \sum_{j=1}^{k-1} \Psi(k-j)\Psi(k+j)\left(\frac{k}{j}\right)^{\tfrac{1}{2}}\leq 0.78.
\]
Moreover, since $\Psi(j+k+1)\leq \Psi(j+k)$ and $(2k+1)/(2j+1)\leq 2k/j$ for $j\geq k+1$, we have from~\eqref{eq:3}
\[
 (ii)\leq \frac{\sqrt{2}}{\pi} \sum_{j=k+1}^{\infty} \Psi(j-k)\Psi(j+k)\left(\frac{k}{j}\right)^{\tfrac{1}{2}} \leq 0.95\times\sqrt{2} \leq 1.35.
\]
Since $|A_{ii}|\leq A_{33} \leq 5/16$ for $i\geq 3$ and 
$5/16 + 0.78 + 1.35\leq 2.45$, these Gerschgorin disks are 
contained in $\{z\in\mathbb{C} : |z|\leq 2.45\}$. 

By Theorem~\ref{thm:Gerschgorin} we conclude that $\max_{1\leq i\leq M+1}|\lambda_{i}(S^{+})|\leq 2.45< 5/2$.
\end{proof}

\end{document}